\documentclass[10pt,a4paper,oneside]{article}

\usepackage{amssymb}
\usepackage{array}

\usepackage{amsmath}
\usepackage{amsthm}
\usepackage{algorithm}
\usepackage[noend]{algorithmic}
\usepackage{listings}
\usepackage{enumerate}
\usepackage{multirow}
\usepackage{graphicx}
\usepackage{url}
\usepackage{color}
\usepackage{mathpartir}
\usepackage{stmaryrd}
\usepackage{booktabs}
\usepackage{xcolor}
\usepackage{multicol}
\usepackage[caption = false]{subfig}
\usepackage{fixltx2e}
\usepackage[normalem]{ulem}
\usepackage{lineno, blindtext}
\usepackage{multirow}
\usepackage{slashbox}
\usepackage{authblk}

\newif\iffull
\fulltrue

\newtheorem{lemma}{Lemma}[section]
\newtheorem{definition}{Definition}[section]
\newtheorem{theorem}{Theorem}[section]
\newtheorem{proposition}{Proposition}[section]

\def\execution#1#2#3{\ensuremath{(#1, #2)\Downarrow#3}}

\def\executionnt#1#2{\ensuremath{(#1, #2) \rightarrowtriangle \dots}}

\def\defseq#1{\ensuremath{#1 = (\defvec)^*}}
\def\loweq#1#2{\ensuremath{\restrict{#1}{L} = \restrict{#2}{L}}}

\def\Irck#1#2{\ensuremath{I_{rc_{#2}}^{#1}}}
\def\Irc#1{\ensuremath{I_{rc}^{#1}}}

\def\channeleq#1#2#3{\ensuremath{\restrict{#1}{#3} = \restrict{#2}{#3}}}

\def\dequeue#1#2{\ensuremath{dequeue(#1, #2)}}

\def\restrict#1#2{\ensuremath{{#1}|_{#2}}}

\def\constructV#1#2{\ensuremath{\ouvd[#2\mapsto#1]}}
\def\substitute#1#2#3{\ensuremath{#1[#2\mapsto#3]}}

\def\select#1{\ensuremath{pick(#1)}}

\def\remove#1#2{\ensuremath{remove(#1,#2)}}

\def\fork#1#2{\ensuremath{fork(#1,#2)}}

\def\fassign#1#2#3#4#5#6{\ensuremath{assign(#1,#2,#3,#4,#5,#6)}}

\def\signal{\ensuremath{signal}}

\def\intsig#1{\ensuremath{#1}}
\def\ProgM{\ensuremath{\Prog_{M}}}
\def\ProgR{\ensuremath{\Prog_{R}}}
\def\st{\ensuremath{st}}
\def\sE{\textbf{E}}
\def\sS{\textbf{S}}

\def\Vals{\ensuremath{\Sigma}}

\def\LVL{\ensuremath{LVL}}
\def\TAV{\ensuremath{T_{M}}}
\def\TPV{\ensuremath{T_{R}}}
\def\ST{{\ensuremath{EX}}}
\def\TOP{\ensuremath{TOP}}

\def\ouvd{\ensuremath{\vec{\bot}}}

\def\VTRUE{\ensuremath{\textbf{T}}}
\def\VFALSE{\ensuremath{\textbf{F}}}

\def\chnl{\ensuremath{c}}

\def\NIL{\ensuremath{\bot}}
\def\Cin{\ensuremath{C_{in}}}
\def\Cout{\ensuremath{C_{out}}}

\def\ACCORCIA{\vspace*{-\baselineskip}}

\def\lpair#1#2{#1\!\!:\!\!\ensuremath{#2}}
\def\lcomma{\ensuremath{\!\!:\!\!}}
\def\lecconf#1#2#3{\ensuremath{\ST[#1].#2\lcomma#3}}

\def\sanserif#1{\ensuremath{\sf #1}}
\def\REDUCE{\ensuremath{\sanserif{REDUCE}}}
\def\MAP{\ensuremath{\sanserif{MAP}}}

\def\defVal{\ensuremath{val_{def}}}
\def\valueM{\ensuremath{val}}
\def\gconf{\ensuremath{\gamma}}
\def\emptyQ{\ensuremath{\epsilon}}
\def\defvec{\ensuremath{\vec{df}}}

\def\EM{\ensuremath{{\sanserif{EM}}}}
\def\Prog{\ensuremath{\pi}}
\def\Progl#1{\ensuremath{\Prog[#1]}}
\def\comm{\ensuremath{\pi}}
\def\commM{\ensuremath{\comm_M}}
\def\commR{\ensuremath{\comm_R}}

\def\Iid{\ensuremath{I}}
\def\Oid{\ensuremath{O}}
\def\EMP{\ensuremath{\EM(}\ensuremath{\Prog)}}

\def\prop1{\text{SubDI}}

\def\EmptyRule#1#2{\ensuremath{\inferrule*[Left={}]{#1}{#2}}}
\def\LabelRule#1#2#3{\ensuremath{{\inferrule*[Left={#1}]{#2}{#3}}}}
\def\ruleconl#1#2{\ensuremath{\Delta, #1  \rightarrowtriangle \Delta, #2}}
\def\RASSG{ASSG}
\def\RCOMP{COMP}
\def\RIFT{IF-T}
\def\RIFF{IF-F}
\def\RWHILET{WHIL-T}
\def\RWHILEF{WHIL-F}
\def\RSKIP{SKIP}
\def\RINPUT{INP}
\def\ROUTPUT{OUTP}

\def\RINPUTYL{LINP1}
\def\RINPUTNL{LINP2}
\def\ROUTPUTL{LOUTP}

\def\RINPM{INPM}
\def\RMAP{MAP}
\def\RCLONE{CLON}

\def\RWAKE{WAKM}

\def\RWAKER{WAKR}
\def\RINPUTR{RETR}
\def\ROUTPUTR{OUTR}
\def\RCLEAN{CLN}

\def\RMAPINIT{MACT}

\def\RREDUCEINIT{RACT}

\def\NASSG{assignment}
\def\NCOMP{sequence}
\def\NIF{if}
\def\NWHILE{while}
\def\NSKIP{skip}
\def\NINPUT{input}
\def\NOUTPUT{output}
\def\NMAP{map}
\def\NCLONE{clone}
\def\NWAKE{wake}
\def\NINPUTR{retrieve}

\def\NCLEAN{clean}

\def\iinput#1#2{\textbf{input}~\ensuremath{#1}~\textbf{from}~\ensuremath{#2}}

\def\iinputr#1#2#3{\textbf{retrieve}~\ensuremath{#1}~\textbf{from}~(#2,#3)}
\def\ioutput#1#2{\textbf{output} ~\ensuremath{#1}~ \textbf{to}~ \ensuremath{#2}}
\def\imap#1#2#3{\ensuremath{\textbf{\NMAP}(#1, #2,#3)}}
\def\iclone#1#2#3{\ensuremath{\textbf{\NCLONE}(#1,#2,#3)}}

\def\iskip{\ensuremath{\textbf{\NSKIP}}}
\def\iwake#1{\ensuremath{\textbf{\NWAKE}(#1)}}
\def\iwaker#1{\ensuremath{\textbf{\NWAKE}(#1)}}
\def\iclean#1#2{\ensuremath{\textbf{\NCLEAN}(#1,#2)}}
\def\ctab{}
\def\icomment#1{\ensuremath{\ctab/\!/\textnormal{#1}}}
\def\qItem#1#2{\ensuremath{(\linecode{#1}=#2)}}
\def\spaceIns{~~~~~~~~~~}
\def\NPRED{\ensuremath{PRED}}
\def\NINIMAP{\ensuremath{WAITI}}
\def\NINIREDUCE{\ensuremath{WAITO}}

\def\NCMAP{\ensuremath{canTell}}

\def\NISREADY{\ensuremath{isReady}}

\def\emptyPlace{\ensuremath{[\ ]}}

\def\canMap#1{\ensuremath{\NCMAP(#1)}}
\def\isReady#1{\ensuremath{\NISREADY(#1)}}

\def\identical#1{\ensuremath{identical(#1)}}
\def\PREDD{\ensuremath{\NPRED\emptyPlace}}
\def\INIMAPD{\ensuremath{\NINIMAP\emptyPlace}}
\def\INIREDUCED{\ensuremath{\NINIREDUCE\emptyPlace}}

\def\PREDE#1{\ensuremath{\NPRED[#1]}}
\def\INIMAPE#1{\ensuremath{\NINIMAP[#1]}}
\def\INIREDUCEE#1{\ensuremath{\NINIREDUCE[#1]}}

\def\Pred#1{\ensuremath{Pred(#1)}}

\def\indexOf#1{\ensuremath{assignIndex(#1)}}
\def\COND{\ensuremath{COND[\ ]}}
\def\CONDP#1{\ensuremath{COND[#1]}}

\def\lprog{{\sf prg}}
\def\lprogm{{\sf prg}}
\def\lprogr{{\sf prg}}

\def\lmem{{\sf mem}}
\def\lmemi{{\sf \ensuremath{mem}}}
\def\lmemm{{\sf \ensuremath{mem}}}
\def\lmemr{{\sf \ensuremath{mem}}}
\def\linput{{\sf in}}
\def\linputi{{\sf in}}
\def\loutput{{\sf out}}
\def\loutputi{{\sf out}}
\def\ltav{\ensuremath{{\sf t_{m}}}}
\def\ltpv{\ensuremath{{\sf t_{r}}}}
\def\ltop{{\sf top}}
\def\LECS{{\sf \ensuremath{LECS}}}
\def\LECSI{\ensuremath{\LECS_{i}}}

\def\LECSUIR#1#2{\ensuremath{\bigcup_{#1}#2}}
\def\litr{{\sf int}}

\def\lmap{{\sf map}}
\def\lreduce{{\sf red}}
\def\lstate{{\sf stt}}
\def\CMAP{{\ensuremath{M}}}
\def\CREDUCE{{\ensuremath{R}}}

\def\tgetask{\ensuremath{at}}
\def\task{\ensuremath{a}}
\def\tget{\ensuremath{t}}
\def\tnoaction{\ensuremath{-}}
\def\ttellput{\ensuremath{at}}
\def\ttell{\ensuremath{t}}
\def\tput{\ensuremath{a}}
\def\tindex#1{\ensuremath{\Progl{#1}}}

\def\tcolm{\ensuremath{PRIV_{T_M}}}
\def\tcolr{\ensuremath{PRIV_{T_R}}}

\def\length#1{\ensuremath{\parallel#1\parallel}}

\def\figdesc#1{
\noindent
\centering
\begin{minipage}{0.95\columnwidth}
\vspace{3pt}
\begin{footnotesize}
#1
\end{footnotesize}
\end{minipage}
}

\lstdefinelanguage{javascript}
{morekeywords={abstract,boolean,break,byte,case,catch,char,class,const,continue,debugger,default,delete,do,double,else,enum,export,extends,false,final,finally,float,for,function,goto,if,implements,import,in,	instanceof,int,interface,long,native,new,null,package,private,protected,public,return,short,static,super,switch,synchronized,this,throw,throws,transient,true,try,typeof,var,void,volatile,while,with, input, output, from, to, then},
showstringspaces=false,
sensitive,
morecomment=[l]//,
morecomment=[s]{/*}{*/},
morestring=[b]",
morestring=[b]'}[keywords,comments,strings]
\lstnewenvironment{javascript}
 {\lstset{language=javascript,
          xleftmargin=9pt,
          numbers=left,
          numbersep=3pt,
          numberstyle=\tiny,
          frame = single,
          basicstyle=\small\tt,
          tabsize=2,
          columns=fullflexible,
          captionpos=b,
          keywordstyle=\bfseries,
          mathescape = true,
          commentstyle=\sl\ttfamily,
          escapechar=\%,
          belowskip = 0pt,
          framexleftmargin = 5pt,
          escapeinside={(*}{*)}
          }}
 {}

\newfont{\pica}{cmpica scaled 800}
\def\linecode#1{{\texttt{#1}}}
\def\linecodeb#1{{\texttt{\bfseries #1}}}

\newsavebox{\mylistingbox}

\newif\ifquestions
\questionstrue 

\newif\ifremarks
\remarkstrue 

\newif\ifskeletons
\skeletonstrue 

\newif\ifideas
\ideastrue 

\newcolumntype{m}{>{$\vcenter\bgroup\hbox\bgroup}c<{\egroup\egroup$}}
\skeletonsfalse  

\begin{document}


\title{MAP-REDUCE Runtime Enforcement of Information Flow Policies}


%

\author{Minh Ngo}
\author{Fabio Massacci}
\author{Olga Gadyatskaya}
\affil{University of Trento, Italy\\ \small{\{surname\}@disi.unitn.it} }

\date{April 2013\\Revised May 2013}

\maketitle

\pagestyle{plain}     
\enlargethispage{4ex} 

\begin{abstract}
We propose a flexible framework that can be easily customized to enforce a large variety of information flow properties. Our framework combines the ideas of secure multi-execution and map-reduce computations. The information flow property of choice can be obtained by simply changes to a map (or reduce) program that control parallel executions.

We present the architecture of the enforcement mechanism and its customizations for non-interference (NI) (from Devriese and Piessens) and some properties proposed by Mantel, such as removal of inputs (RI) and deletion of inputs (DI), and demonstrate formally soundness and precision of enforcement for these properties.
\end{abstract}

\section{Introduction}\label{sec:introduction}
Information flow properties define the acceptable behaviours of computer programs with respect to allowed and forbidden flows of information. The most well-known information flow property is \emph{non-interference} (NI), which roughly requires that the input data classified as confidential (also called secret, or high) should not influence the public (low) outputs \cite{Gogu-Mese-82-IEEESP,Devr-Pies-10-IEEESP}. 

By weakening or strengthening the definition of NI in order to address some of its problems, security researchers have proposed different information flow properties \cite{MANT-00-CSF,McCu-87-SP,Guttman-Nadel-88-CSF,McLe-94-SP,ZAKI-LEE-97-SP}. For instance, the definition of NI in \cite{Gogu-Mese-82-IEEESP} is based on an assumption that if there is no high input, then there is no high output. This assumption does not always hold. In \cite{McLe-94-SP}, the \emph{generalized non-inference} (GNF) property is defined for systems that generate high outputs even if there are no high inputs.

Different information flow properties led to different enforcement techniques. To the best of our knowledge, there is no proposal in the literature with a unified approach to the enforcement of multiple information flow properties. The existing enforcement mechanisms (e.g. \cite{Devr-Pies-10-IEEESP,Bart-Arge-Rezk-11-MSCS,Volpano-Irvine-Smith-96-JCS,LeGu-07,Shro-Smit-Thob-2007,Russo-Sabe-09-ESORICS,Capi-Long-Venk-08-ACSAC}) can be configured to accommodate different information flow policies that identify what is confidential and what is public, and what are the authorized flows in the security lattice \cite{Devr-Pies-10-IEEESP,Sabe-Myer-2003}, and, sometimes, they can as well enforce declassification policies
\footnote{These policies are required when one needs to disclose information that depends on confidential data in some way, see e.g. \cite{Myers-Sabelfeld-Zdancewic-04-CSFW,Sabelfeld-Sands-09-JCS} for details.}
(e.g. \cite{Aust-Flan-12-POPL}). Yet, the adaptation of an existing enforcement mechanism (for example, for NI) to enforce another property (for instance, GNF) is not straight-forward.

We aim to fill this gap by providing an enforcement framework that can be extended by different information flow properties.  The framework is inspired by the MAP-REDUCE approach explored by Google \cite{Lamm-07}; and generalizes the secure multi-execution (SME) technique proposed by Devriese and Piessens in \cite{Devr-Pies-10-IEEESP} so that it can enforce other information flow properties, e.g. properties from \cite{MANT-00-CSF}.

The main idea is to execute multiple ``local'' instances of the original program, feeding different inputs to each instance of the program. The local inputs are produced from the original program inputs by the \MAP\ component, depending on the set of security levels defined in the framework and the input channels available. Upon receiving the necessary data (for instance, after each individual program instance is terminated), the \REDUCE\ component collects the local outputs and generates the common output, thus ensuring that the overall execution is secure. \MAP\ and \REDUCE\ are customizable and by changing their programs the user can easily change the enforced property. Two simple tables (\TAV\ and \TPV) tell \MAP\ and \REDUCE\ what they should do when receiving respectively input and output requests from local executions on a channel.


In this report we present the following contributions:

\begin{itemize}
\item The architecture of this flexible enforcement framework.

\item A set of simple instructions for programming the framework components, such as a ``\NCLONE'' instruction to spawn new processes.

\item The instantiation of the framework's configuration for non-interference (NI) from \cite{Devr-Pies-10-IEEESP}, Removal of Inputs (RI) and Deletion of Inputs (DI) from \cite{MANT-00-CSF}. The components are summarized in Tab.~\ref{tab:component:ifp}. We prove formally soundness and precision of these enforcement mechanisms with respect to the corresponding properties for a model programming language with simple I/O instructions.
\item An example on how a simple change to the configuration can lead to the enforcement of new information flow properties.
\end{itemize}

\begin{table}
\centering
\caption{Enforcement mechanisms for the selected information flow properties}
\label{tab:component:ifp}
\begin{tabular}{|p{4.1cm}|c|c|c|c|}
\hline
\multicolumn{1}{|c|}{\multirow{2}{*}{\textbf{Property}}} & \multirow{2}{*}{\textbf{Section}} & \multicolumn{3}{c|}{\textbf{Components}}\\
\cline{3-5}
 & & \MAP& \REDUCE& \TAV/\TPV \\
\hline
\hline
Removal of inputs \cite{MANT-00-CSF} & \S\ref{sec:em:RI} & Fig.\ref{alg:MAP:RI} & Fig.\ref{alg:REDUCE:RI} & Fig.\ref{fig:table:RI:MAP},\ref{fig:table:RI:REDUCE} \\
\hline
Deletion of inputs \cite{MANT-00-CSF}  & \S\ref{sec:em:DI} & Fig.\ref{alg:MAP:DI} & Fig.\ref{alg:REDUCE:DI} & Fig.\ref{fig:table:DI:MAP},\ref{fig:table:DI:REDUCE} \\
\hline
Termination (in)sensitive  & \S\ref{sec:em:NI} & Fig.\ref{alg:MAP:NI:SME} & Fig.\ref{alg:REDUCE:NI} & Fig.\ref{fig:table:NI:MAP},\ref{fig:table:NI:REDUCE} \\
non-interference \cite{Devr-Pies-10-IEEESP} & & & &\\
\hline
\end{tabular}
\end{table}

The rest of the paper is organized as follows. \S\ref{sec:overview} gives an overview of the idea behind our approach and the architecture of the enforcement framework. The semantics of the controlled programs is introduced in \S\ref{sec:formalization:proram}. The formalization of the framework is presented in \S\ref{sec:formalization:em}. The enforcement mechanisms for the information flow properties we have selected are described in \S\ref{sec:em}. The soundness and precision of the enforcement mechanisms constructed are presented in respectively \S\ref{sec:soundness} and \S\ref{sec:precision}. We discuss options for fine-tuning the enforcement framework and ideas for its further extensions with other properties  in \S\ref{sec:discussion}. The relationships among the properties enforced, and the limitations of the framework are discussed respectively in \S\ref{sec:discussion:relationship} and \S\ref{sec:discussion:limitations}. Then we discuss related work in \S\ref{sec:relwork} and conclude in \S\ref{sec:conclusion}.

\section{Overview}\label{sec:overview}
Fig.~\ref{fig:archiectureEMP} depicts the general architecture of the enforcement mechanism for an information flow property on a program \Prog. It is composed by a stack \ST\ of local executions ($\Prog[0], \dots,$ $\Prog[\TOP]$, where $\TOP$ is the index of the top of the stack), global input and output queues, the \MAP\ and \REDUCE\ components, and the tables \TAV\ and \TPV.

Local executions (instances of the original program that are executed in parallel and are unaware of each other) are separated from the environment input and output actions by the enforcement mechanism. A local execution has its own input and output queues. The local input (resp. output) queue of a local execution contains the input (resp. output) items that can be freely consumed (resp. generated) by this local execution. \MAP\ and \REDUCE\ are responsible for respectively the global input queue containing the input items from the external environment (received from the user or other input channels), and the global output queue containing the output items filtered by the enforcement mechanism to the environment.

\begin{figure}
\centering
\includegraphics[scale=0.75]{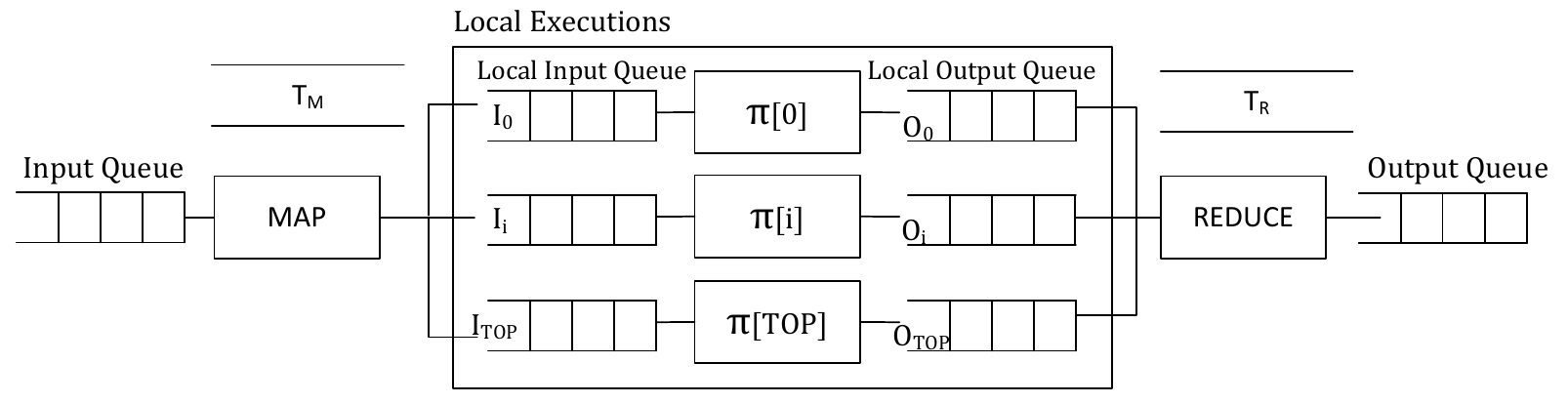}
\caption{Architecture of enforcement mechanisms}
\label{fig:archiectureEMP}
\end{figure}

When a local execution needs an input item that is not yet ready in its local input queue, it will request the help of \MAP\ by emitting an \emph{interrupt signal} (or just \emph{signal} for short). When different local executions request values from the same channel, there will be only one actual input action performed by the enforcement mechanism. After the value is read, \MAP\ will distribute it to local executions, replacing the actual value by the default (fake) one, if necessary. Similarly, when a local execution generates an output item, the output item will be handled by \REDUCE.

\MAP\ and \REDUCE\ can also autonomously send and, respectively, collect items from local queues. For example, upon receiving an input item from the environment, \MAP\ can send it to \emph{all} local executions that satisfy a predicate. The parallel broadcast and parallel collection to and from local processors are the characteristic features of \MAP-\REDUCE\ programs \cite{Lamm-07}; this explains our choice for the name of the enforcement mechanism.

The actions of \MAP\ (respectively \REDUCE) on an input (output) request from a local execution depend on the configuration information in the table \TAV\ (\TPV). These components of the enforcement mechanism are customized depending on the desired information flow property. The framework components configured to implement the chosen information flow properties are listed in Tab.~\ref{tab:component:ifp} (for each selected property the table contains pointers to the actual component configurations).

The configuration of input and output actions of local executions is based on two privileges: \emph{ask} (\task) and \emph{tell} (\ttell). If a local execution has the \emph{ask} privilege on the input channel \chnl, then \MAP\ can fetch the input item from the environment upon receiving the interrupt signal from a local execution. If a local execution has the \emph{tell} privilege on the input channel \chnl, then this local execution can get the real value from the channel \chnl\ when \MAP\ broadcasts the input item to local executions, otherwise it will get a default value. If a local execution has the \emph{ask} privilege on the output channel \chnl, then \REDUCE\ will actually ask the execution for the real value that it wished to send to \chnl. Otherwise, \REDUCE\ will just replace it with a default value. If a local execution has the \emph{tell} privilege on the output channel \chnl, it can invoke \REDUCE\ to send the values generated by itself to \chnl.

Notice that an execution may have only one privilege. For example, an execution with the \emph{ask} but not the \emph{tell} privilege in \TPV\ will provide the real value to \REDUCE, but will not be able to invoke \REDUCE\ to put the value in the external output. It will have to wait for somebody else with the \emph{tell} privilege on the channel to produce an output.

\section{Semantics of Controlled Programs}\label{sec:formalization:proram}
\begin{figure}
\centering
	\begin{alignat*}{2}
	\comm ::= &~& \spaceIns instructions:\\
		& |x:=e & \NASSG\\
		& |\comm;\comm & \NCOMP\\
		& |\textbf{if} ~e ~\textbf{then} ~\comm ~\textbf{else} ~\comm & \NIF\\
		& |\textbf{while} ~e ~\textbf{do} ~\comm & \NWHILE \\
		& |\textbf{skip} & \NSKIP\\
		& |\textbf{input} ~x ~\textbf{from} ~c & \NINPUT\\
		& |\textbf{output} ~e ~\textbf{to} ~c & \NOUTPUT
\end{alignat*}
\caption{Language instructions}
\label{fig:Comm:Standard}
\end{figure}

Our model programming language is close to the one used in the SME paper \cite{Devr-Pies-10-IEEESP}. Valid values in this language are boolean values ($\VTRUE$ and $\VFALSE$) or non-negative integers. A program \Prog\ is an instruction composed from the terms described in Fig.~\ref{fig:Comm:Standard}. In this figure, $\comm$, $e$, $x$, and $\chnl$ are meta-variables for respectively instructions, expressions, variables, and input/output channels. Since a program is just a sequence of instructions (i.e. a complex instruction itself), we will use program and instruction interchangeably when referring to complex instructions.

We model an input (output) item as a vector and define input (output) of program instances as queues. We use vectors of channel to accommodate forms in which multiple fields are submitted simultaneously but are classified differently (e.g. credit card numbers vs. user names)

\begin{definition}\label{def:vector}
An \emph{input vector} $\vec{v}$ is a mapping from input channels to their values, $\vec{v}: C_{in}  \rightarrow \Vals \cup \{\NIL\}$, where the value \NIL\ is the special undefined value. An \emph{output vector} $\vec{v}$ is a mapping from output channels to their values, $\vec{v}: C_{out} \rightarrow \Vals \cup \{\NIL\}$.
\end{definition}

Given a vector $\vec{v}$ and a channel $\chnl$, the \emph{value of the channel} is denoted by $\vec{v}[\chnl]$. The symbol $\ouvd$ denotes an output vector mapping all channels to \NIL. To simplify the formal presentation, in the sequel w.l.o.g. we assume that each input and output operation only affect one channel at a time. Thus, for each vector, there is only one channel $\chnl$ such that $\vec{v}[\chnl] \neq \NIL$.

Let queue $Q$ be a sequence of elements $q_1 \dots q_n$. We denote the addition of a new element to the queue $Q$ as $Q.q$, or $q_1 \dots q_n.q$; the removal of the first element from the queue $Q$ is denoted by $q_2 \dots q_n$. By $\emptyQ$ we denote an empty queue.

To define an execution configuration, we use a set of labelled pairs. A labelled pair is composed by a label and an object and in the form ${\sf label}\lcomma object$. The ${\sf label}$ is attached to the $object$ in order to differentiate this object from others, so each label occurs only once. For example, $\ltav:\TAV$ is the configuration table for \MAP. A summary of labels and their semantics used in this report is in Tab.~\ref{tbl:label-semantics}.


\begin{table}
\centering
\caption{Labels and their semantics}
\label{tbl:label-semantics}
\begin{tabular}{|c|l|}
\hline
\textbf{Label} & \multicolumn{1}{c|}{\textbf{Semantics}} \\
\hline
\hline
\lprog & Program executed by the component \\
\hline
\lmem &  Memory of the component \\
\hline
\linput & Input queue of the component \\
\hline
\loutput & Output queue of the component \\
\hline
\ltop & Index of the top of the stack \ST \\
\hline
\lstate & State of a local execution \\
\hline
\litr & Interrupt signal sent by a local execution\\
\hline
\ltav & Table \TAV \\
\hline
\ltpv & Table \TPV \\
\hline
\lmap & \MAP\ component \\
\hline
\lreduce & \REDUCE\ component\\
\hline
\end{tabular}
\end{table}

\begin{definition}\label{def:execConf}
An \emph{execution configuration} of a program is a set $\{\lpair{\lprog}{\comm}, \lpair{\lmem}{m}, \lpair{\linput}{I}, \lpair{\loutput}{O} \}$, where $\comm$ is the program to be executed,  $m$ is the memory (a function mapping variables to values), $I$ ($O$ respectively) is the queue of input (resp. output) vectors.
\end{definition}

\begin{figure}[!t]
\begin{center}
\ACCORCIA
\begin{mathpar}
		\LabelRule{\RASSG}{\comm = x:= e \\ m(e) = \valueM}
		{\ruleconl {\lpair{\lprog}{\comm}, \lpair{\lmem}{m}}
		{\lpair{\lprog}{\textbf{skip}}, \lpair{\lmem}{\substitute{m}{x}{\valueM}}}} \\
		\LabelRule{\RCOMP}{\lprog \lcomma \comm_1, \lmem \lcomma m, \linput \lcomma I, \loutput \lcomma O	\rightarrowtriangle \lprog \lcomma \comm_1', \lmem \lcomma m', \linput \lcomma I', \loutput \lcomma O'}
		{\lprog\lcomma\comm_1;\comm_2, \lmem\lcomma m, \linput\lcomma I, \loutput\lcomma O  \rightarrowtriangle \lprog\lcomma\comm_1';\comm_2,\lmem\lcomma m', \linput\lcomma I', \loutput\lcomma O'} \\
		\LabelRule{\RIFT}{\comm = \textbf{if}~e~\textbf{then}~\comm_1 ~\textbf{else}~ \comm_2\\
		 m(e) = \VTRUE}
		{ \ruleconl{\lprog\lcomma \comm}{\lprog\lcomma\comm_1}} \\
		\LabelRule{\RIFF}{\comm = \textbf{if}~e~\textbf{then}~\comm_1 ~\textbf{else}~ \comm_2\\
		 m(e) = \VFALSE}
		{\ruleconl{ \lprog\lcomma\comm}{\lprog\lcomma\comm_2}}  \\
		\LabelRule{\RWHILET}{\comm = \textbf{while}~e ~\textbf{do}~\comm_{loop}\\
		 m(e) = \VTRUE }
		{\ruleconl{\lprog\lcomma\comm}{\lprog\lcomma\comm_{loop};\comm}} \\
		\LabelRule{\RWHILEF}{\comm = \textbf{while}~e ~\textbf{do}~\comm_{loop}\\
		 m(e) = \VFALSE }
		{\Delta, \lprog\lcomma\comm \rightarrowtriangle \Delta,\lprog\lcomma\textbf{skip}}\\
\LabelRule{\RSKIP}{ }
		{\ruleconl{\lprog\lcomma\textbf{skip};\comm}{ \lprog\lcomma\comm}}\\
		\LabelRule{\RINPUT}{\comm = \textbf{input} ~ x ~ \textbf{from} ~ c \\ I = \vec{v}.I'   \\ \vec{v}[c] \neq \NIL}
		{\Delta, \lprog\lcomma\comm, \lmem\lcomma m, \linput\lcomma I \rightarrowtriangle \Delta, \lprog\lcomma\textbf{skip}, \lmem\lcomma \substitute{m}{x}{\vec{v}[c]}, \linput\lcomma I'} \\
		\LabelRule{\ROUTPUT}{\comm = \textbf{output} ~ e ~\textbf{to} ~c \\
		 \vec{v} = \constructV{m(e)}{\chnl}}
		{\Delta, \lprog\lcomma\comm, \loutput \lcomma O  \rightarrowtriangle \Delta, \lprog\lcomma \textbf{skip},\loutput\lcomma O.\vec{v}} \and
\end{mathpar}
\ACCORCIA	
\caption{Semantics of instructions of controlled programs}
\label{fig:CommandSem}
\end{center}
\end{figure}

The operational semantics of the language is described in Fig.~\ref{fig:CommandSem}. The conclusion part of each semantic rule is written as $\Delta, \Gamma \Rightarrow \Delta,\Gamma'$, where $\Delta$ denotes the elements of the execution configuration that are unchanged upon the transition. The semantics of the comma ``$,$'' in the expression $\Delta, \Gamma$ is the disjoint union of $\Delta$ and $\Gamma$. We abuse the notation of the memory function $m(.)$ and use it to evaluate expressions to values. When an output command sends a value to the channel $c$, an output vector $\vec{v} = \constructV{\valueM}{\chnl}$ is inserted into the output queue, where $\vec{v}$ is the vector with all undefined channels, except \chnl\ that is mapped to $m(e)$, so $\vec{v}[c'] = \NIL$ for all $c' \neq c$ and $\vec{v}[c] = m(e)$.

\begin{definition} \label{def:execution}
An \emph{execution} of the program $\Prog$ is a finite sequence of configuration transitions $\gconf_0 \rightarrowtriangle \gconf_1 \rightarrowtriangle \ldots \rightarrowtriangle \gconf_k$, where $\gconf_0 = \{ \lprog\lcomma\Prog, \lmem\lcomma m_0, \linput\lcomma I, \loutput \lcomma \emptyQ \}$ is the initial configuration, $m_0$ is the function mapping every variable to the initial value, and $k$ is the number of transitions.
\end{definition}

The transition sequence can be also written as  $\gconf_0 \rightarrowtriangle^k \gconf_k$, or $\gconf_0 \rightarrowtriangle^* \gconf_k$ if the exact number of transitions does not matter, or $(\Prog$, $I, \emptyQ)$ $\rightarrowtriangle^*$ $(\comm^k, I^k, O^k)$, where $\gconf_k =  \{\lprog\lcomma\comm^k, \lmem\lcomma m^k, \linput \lcomma I^k, \loutput\lcomma O^k \}$. All sequences have the form

\begin{equation*} \label{conf:transition}
\left((\gconf_1 \overset{\EmptyRule{\text{INST}_1}{\text{COMP}}}{\rightarrowtriangle} \gconf_2 \overset{\EmptyRule{}{\text{\RSKIP}}}{ \rightarrowtriangle} \gconf_3)|(\gconf_1 \overset{\EmptyRule{\text{INST}_2}{\text{COMP}}}{\rightarrowtriangle} \gconf_2)|(\gconf_1 \overset{\EmptyRule{}{\text{\RSKIP}}}{\rightarrowtriangle} \gconf_2)\right)^* \overset{\EmptyRule{\text{INST}_3}{\text{\RCOMP}}}{\rightarrowtriangle} \gconf_4
\end{equation*}, where  $\text{INST}_1$ is \RASSG, \RWHILEF, \RINPUT, or \ROUTPUT; $\text{INST}_2$ is \RIFT, \RIFF, or \RWHILET; and $\text{INST}_3$ is \RASSG, \RINPUT, \ROUTPUT, or \RSKIP.


An infinite sequence is written as $\gconf_0 \rightarrowtriangle \gconf_1 \rightarrowtriangle \dots$.

\begin{definition}\label{def:termination}
The program \emph{terminates} if there exists a configuration $\gconf_f = \{ \lprog \lcomma \NSKIP, \lmem\lcomma m, \linput\lcomma \emptyQ, \loutput\lcomma O \}$ such that $\gconf_0 \rightarrowtriangle^* \gconf_f$. We denote this whole derivation sequence by $\execution{\Prog}{I}{O}$ using the big step notation.
\end{definition}

\section{Semantics of the Enforcement Mechanism}\label{sec:formalization:em}
\begin{definition}\label{def:sem:global}
A \emph{configuration of an enforcement mechanism} is a set $\{\ltav\lcomma\TAV, \ltpv\lcomma\TPV, \ltop\lcomma\TOP, \lmap\lcomma\CMAP, \lreduce\lcomma\CREDUCE, \linput\lcomma I, \loutput\lcomma O, \LECSUIR{i}{\LECSI}\}$, where \TAV\ and \TPV\ are configuration tables for respectively \MAP\ and \REDUCE, \TOP\ is the index of the top of the stack of configurations of local executions \ST, $\CMAP$ and $\CREDUCE$ are configurations of respectively \MAP\ and \REDUCE\ components, $I$ and $O$ are respectively the input and output queues of the enforcement mechanism, and \LECSI\ is the configuration of the i-th local execution.
\end{definition}

We denote the enforcement mechanism on \Prog\ by \EMP. For the initial configuration, all local input and output queues will be empty, all local executions will be in the executing state, and skip is the only instruction in \MAP\ and \REDUCE\ programs. The enforcement mechanism is terminated when all local executions, \MAP\ and \REDUCE\ programs are terminated, and the global input queue is consumed completely.

\begin{definition}\label{def:termination:em}
The enforcement mechanism \emph{terminates} if there exists a configuration $\gconf_f = \{\ltav\lcomma\TAV, \ltpv\lcomma\TPV, \ltop\lcomma\TOP, \lmap\lcomma\CMAP, \lreduce\lcomma\CREDUCE, \linput\lcomma \emptyQ, \loutput\lcomma O, \LECSUIR{i}{\LECSI} \}$ such that $\gconf_0 \rightarrowtriangle^* \gconf_f$, where $\ST[i].\lprog\lcomma\iskip$ for all $i$, $\lmap.\lprogm\lcomma\iskip$, and $\lreduce.\lprogr\lcomma\iskip$. We denote this whole derivation sequence by $\execution{\EMP}{I}{O}$ using the big step notation.
\end{definition}

We now specify the semantics of the enforcement mechanism components: local executions, the programs of \MAP\ and \REDUCE. The general approach is that execution of parallel programs is modeled by the interleaving of concurrent atomic instructions \cite{Lamport-83-SCP} so each transition rule either by a local execution, by \MAP, or by \REDUCE\ is a step of the enforcement mechanism as a whole.

\subsection{Local Executions}\label{sec:formalization:locExec}
Each local execution is associated with a unique identifier $i$, that is its number on the stack \ST. A local execution can be in one of the two states: \sE\ (Executing) or \sS\ (Sleeping). Initially, the state of all local executions is \sE. A local execution moves from \sE\ to \sS\ when it has sent an interrupt signal to require an input item that is not ready in its local input queue, or to signal that it has generated an output item. A local execution moves from \sS\ to \sE\ when it is awaken by the \MAP\ component (the input item it required is ready) or by the \REDUCE\ component (its output item is consumed).

\begin{definition}\label{def:execConf:lec}
An execution configuration of a local execution is a set $\LECSI \triangleq \{\lecconf{i}{\lstate}{\st},\lecconf{i}{\litr}{signal}, \lecconf{i}{\lprog}{\comm}, \lecconf{i}{\lmem}{m}, \lecconf{i}{\linputi}{\Iid}, \lecconf{i}{\loutputi}{\Oid}\}$, where \ST\ is the global stack of local execution, $i$ denotes the i-th execution, $\st$ is the state of the local execution, $signal$ is the interrupt signal sent by the local execution, $\comm$ is an instruction to be executed, $m$ is the memory, and $\Iid$ and $\Oid$ are queues of input and output vectors respectively.
\end{definition}

We define the dequeue operator \dequeue{Q}{\chnl} on a queue $Q$ and a channel \chnl\ that returns $(\valueM, Q')$, where the value of $\valueM$ is $\vec{v}[\chnl]$,  where $\vec{v}$ is the first vector, such that $\vec{v}[\chnl] \neq \NIL$, and $Q'$ is obtained by removing $\vec{v}$ from $Q$; otherwise (there is no vector $\vec{v}$ in $Q$, where $\vec{v}[\chnl] \neq \NIL$), $\valueM = \NIL$  and $Q' = Q$.

The semantics of local executions for assignment, composition, if, while, skip instructions is essentially identical to the one described in Fig.~\ref{fig:CommandSem}. The only difference is the explicit condition that the local state must be $\sE$. We do not present these rules in the paper. We provide the rules for input and output instructions in Fig.~\ref{fig:LocCommandSem2}. When the \NINPUT\ instruction is executed and the input item required is in the local input queue, this item will be consumed (rule \RINPUTYL). Otherwise, the local execution emits an input interrupt signal $\intsig{\chnl}$ and moves to the sleep state (rule \RINPUTNL).  The output interrupt signal $\intsig{\chnl}$ is generated when the \NOUTPUT\ instruction is executed (rule \ROUTPUTL).

\begin{figure}[!t]
\begin{center}
\ACCORCIA
\begin{mathpar}
	\LabelRule{\RINPUTYL}{\ST[i].\st = \sE \\ \comm = \textbf{input} ~ x ~ \textbf{from} ~ c \\ \dequeue{\Iid}{\chnl} = (\valueM,\Iid') \\ \valueM \neq \NIL}
	{\Delta, \ST[i].\lprog\lcomma\comm , \ST[i].\lmemi\lcomma m, \ST[i].\linputi\lcomma \Iid \\\\ \Rightarrow \Delta, \ST[i].\lprog\lcomma\textbf{skip}, \ST[i].\lmemi\lcomma \substitute{m}{x}{\valueM}, \ST[i].\linputi\lcomma \Iid'} \and
	\LabelRule{\RINPUTNL}{\ST[i].\st = \sE \\
		\comm = \textbf{input} ~ x ~ \textbf{from} ~ c \\
		\dequeue{\Iid}{\chnl} = (\NIL,\Iid')}
	{\Delta, \ST[i].\lstate\lcomma\sE, \ST[i].\litr\lcomma \NIL  \Rightarrow \Delta,\ST[i].\lstate\lcomma\sS, \ST[i].\litr\lcomma\intsig{\chnl}} \and
	\LabelRule{\ROUTPUTL}{\ST[i].\st = \sE \\
	\comm = \textbf{output} ~ e ~\textbf{to} ~c \\
	\ST[i].\lmem = m \\
	 \vec{v} = \constructV{m(e)}{\chnl} }
	{\Delta, \ST[i].\lstate\lcomma\sE, \ST[i].\litr\lcomma\NIL, \ST[i].\lprog\lcomma \comm,  \ST[i].\loutputi \lcomma \Oid \\\\ \Rightarrow \Delta,\ST[i].\lstate\lcomma\sS, \ST[i].\litr\lcomma\intsig{\chnl}, \ST[i].\lprog\lcomma\textbf{skip}, \ST[i].\loutputi\lcomma \Oid.\vec{v}}
\end{mathpar}
\ACCORCIA
\caption{Semantics of the \NINPUT\ and \NOUTPUT\ instructions of \Progl{i}}
    \label{fig:LocCommandSem2}
\end{center}
\end{figure}

The initial configuration of the i-th local execution is $\{ \lecconf{i}{\lstate}{\sE},\lecconf{i}{\litr}{\NIL}, \lecconf{i}{\lprog}{\comm}, \lecconf{i}{\lmem}{m_0}, \lecconf{i}{\linputi}{\emptyQ}, \lecconf{i}{\loutputi}{\emptyQ}\}$.
A local execution is terminated if there is only the skip instruction to be executed.

\subsection{\MAP}
A \MAP\ program is normally composed of three steps: the input retrieval step, the value distribution step and the wake up step. In the first step, an input item is fetched by performing an actual input action from the specified channel, or by using the default value (\defVal). In the second step, a real input item or the default item is sent to local executions. These two steps depend on the configuration in \TAV. In the third step, local executions are waken up if a certain condition is satisfied, e.g., these local executions were waiting for input items and they have received the input items they required.

In addition to the instructions in Fig.~\ref{fig:Comm:Standard} (except the \NOUTPUT\ instruction that is replaced by the \NMAP\ instruction), the program \ProgM\ is also composed by the instructions described in Fig.~\ref{fig:Comm:MAP}, where $\PREDD \triangleq \lambda x. \Pred{x}$\ is a meta-variable for predicates. The evaluation of the predicate \PREDE{\ } on the configuration of the local execution \Progl{i} is denoted as \PREDE{i}.

\begin{figure}
\begin{center}
\begin{alignat*}{2}
		\commM ::=&\dots &  \spaceIns instructions:\\
		& |\imap{e}{\chnl}{\PREDD} & \NMAP\\
		& |\iwaker{\PREDD} & \NWAKE \\
		& |\iclone{\PREDD}{\tcolm}{\tcolr} & \NCLONE
\end{alignat*}
\vspace{-10pt}
\caption{\MAP\ instructions}
\label{fig:Comm:MAP}
\end{center}
\end{figure}

\emph{The execution of \NMAP, \NWAKE, or \NCLONE\ instruction is applied simultaneously to all local executions \Progl{i} such that \PREDE{i} is true} as follows. First, the value of the expression $e$ is sent to the input queues of all local executions. The value sent is considered as a value from the channel \chnl. Then all local executions $\Prog[i]$ are awaken and the interrupt signals generated by those local executions (if there were some) are removed. The execution of the \NCLONE\ instruction clones the configuration of each local execution \Progl{i}. The new program and the overall configuration will be appended to the local executions stack. The state of the new executions is \sS. The privileges of the new local executions are copied from the lists of privileges \tcolm\ and \tcolr. The list \tcolm\ (respectively \tcolr) is an input (resp. output) privilege configuration template which varies depending on the enforced property. We give an example of such templates in \S\ref{sec:em:DI}, where the enforced property requires cloning.

\begin{definition}\label{def:confMAP}
A configuration of the \MAP\ component is a set $\{\lmap.\lpair{\lprogm}{\commM}, \lmap.\lpair{\lmem}{m}\}$, where $\comm_M$ is the instruction to be executed, and $m$ is the memory.
\end{definition}

The semantics of instructions of \NASSG, \NCOMP, \NIF, \NWHILE, and \NSKIP\ of \MAP\ is almost the same as the semantics presented in Fig.~\ref{fig:CommandSem}. The \NOUTPUT\ instruction is not used in \ProgM. The semantics of \NINPUT\ and \NMAP\ is described in Fig.~\ref{fig:sem:MAP}. For the \NMAP, \NWAKE, and \NCLONE\ instructions, if there is no $i$ such that \PREDE{i} holds, then the execution of these instructions makes all local executions move from their current configurations to themselves.

\begin{figure}[!t]
\begin{center}
\ACCORCIA
\begin{mathpar}
\LabelRule{\RINPM}{\commM = \textbf{input} ~x ~\textbf{from}~$\chnl$ \\ I = \vec{v}.I' \\ \vec{v}[c] \neq \NIL}
{\Delta,  \lmap.\lprogm\lcomma\commM, \lmap.\lmemm\lcomma m, \linput\lcomma I  \Rightarrow \Delta,\lmap.\lprogm\lcomma\textbf{skip}, \lmap.\lmemm\lcomma \substitute{m}{x}{\vec{v}[\chnl]}, \linput\lcomma I'} \and
\LabelRule{\RMAP}{\commM = \textbf{\NMAP}(e,\chnl, \PREDD) \\
	S = \{i \in \{0,\dots,\TOP\}:\PREDE{i}\} \\
	\LECS = \bigcup_{i \in S}\{\ST[i].\linput\lcomma \Iid \}\\
	\vec{v} = \constructV{m(e)}{\chnl} \\
	\LECS' = \bigcup_{i \in S}\{\ST[i].\linput\lcomma \Iid.\vec{v} \}}
{\Delta, \lmap.\lprogm\lcomma\commM, \LECS  \Rightarrow  \Delta,\lmap.\lprogm\lcomma\textbf{skip},\LECS'} \and
\LabelRule{\RWAKE}{\commM = \textbf{\NWAKE}(\PREDD) \\
	S = \{i \in \{0,\dots,\TOP\}:\PREDE{i}\} \\\\
	\LECS = \bigcup_{i \in S} \{ \ST[i].\litr\lcomma\signal, \ST[i].\lstate\lcomma\sS\} \\
	\LECS' = \bigcup_{i \in S} \{\ST[i].\litr\lcomma\NIL, \ST[i].\lstate\lcomma\sE \} }
{\Delta, \lmap.\lprogm\lcomma\commM, \LECS \Rightarrow \Delta,\lmap.\lprogm\lcomma\textbf{skip}, \LECS'} \and
\LabelRule{\RCLONE}{\commM = \iclone{\PREDD}{\tcolm}{\tcolr} \\
	 S = \{i \in \{0,\dots,\TOP\}:\PREDE{i}\} \\
	 \LECS = \bigcup_{i}\LECSI\\
	 \LECS' = \LECS \cup \bigcup_{i \in S}\fork{\LECSI}{\TOP + \indexOf{i}} \\
	 \TOP' = \TOP + |S|\\
	 (\TAV', \TPV') = \fassign{\TAV}{\TPV}{\TOP}{\TOP'}{\tcolm}{\tcolr}}
{\Delta, \ltav\lcomma\TAV, \ltpv\lcomma\TPV, \ltop\lcomma\TOP, \lmap.\lprogm\lcomma\commM,\LECS
	\\\\ \Rightarrow \Delta,  \ltav\lcomma\TAV', \ltpv\lcomma\TPV', \ltop\lcomma\TOP',\lmap.\lprogm\lcomma\textbf{skip},\LECS'}\and
\LabelRule{\RMAPINIT}{\lmap.\lprogm\lcomma\iskip \\
	 S = \{i \in \{0,\dots,\TOP\}:\INIMAPE{i}\}\\
	 S \neq \emptyset \\
	 i = \select{S} \\
	 \ST[i].\lprog = \iinput{x}{\chnl};\comm}
{\Delta, \ST[i].\litr\lcomma\chnl, \lmap.\lprogm\lcomma\iskip, \lmap.\lmemm\lcomma m \\\\
\Rightarrow \Delta, \ST[i].\litr\lcomma\NIL, \lmap.\lprogm \lcomma\ProgM(i,\chnl), \lmap.\lmemm\lcomma m_0}
\end{mathpar}
\ACCORCIA
\caption{Semantics of instructions of \MAP}
	\label{fig:sem:MAP}
\end{center}
\end{figure}

The bijective function $\indexOf{}:S \rightarrow \{1, \dots, |S|\}$ assigns and returns an unique index of the element $i$ in the set $S$ (the index starts from $1$). The function $\fork{\LECSI}{j}$ makes a copy of the local execution \Progl{i}; the new execution can be referred as $\ST[j]$. The function \fassign{\TAV}{\TPV}{\TOP}{\TOP'}{\tcolm}{\tcolr} modifies tables \TAV\ and \TPV\ by adding new columns for the newly cloned processes and the corresponding values for the privileges from \tcolm\ and \tcolr\ for the input and output channels for these processes.

The initial configuration of \MAP\ is $\{\lmap.\lprogm\lcomma\NSKIP, \lmap.\lmem\lcomma m_0 \}$. The execution of \MAP\ is terminated if \NSKIP\ is the only instruction in the \MAP\ program.

We define the predicate  \INIMAPD\ that indicates whether a local execution is waiting for an input  item or not. The function $\select{S}$ returns an element from the non-empty set $S$. The selection of an element in a non-empty set $S$ can be non-deterministic or in the round-robin way.

\begin{eqnarray*}
\INIMAPD &\triangleq & \lambda x. \ST[x].\lstate = \sS \wedge \exists \chnl \in \Cin: \ST[x].\litr = \intsig{\chnl} \wedge \nonumber \\
		& & {}\wedge \ST[x].\lprog = \iinput{y}{\chnl};\comm
\end{eqnarray*}

The \MAP\ component activation is presented in Fig.~\ref{fig:sem:MAP} (the rule $\RMAPINIT$). \MAP\ can be activated when there is only the \NSKIP\ instruction in the \MAP\ program, and there is an interrupt signal  \intsig{\chnl} from the local execution \Progl{i}, the state of this local execution is sleeping (\sS), the instruction to be executed is an input instruction.  The activation of \MAP\  on a signal on channel \chnl\ from \Progl{i} will remove the signal from configuration of \Progl{i}.

\subsection{\REDUCE} \label{sec:emp:reduce}
The \REDUCE\ component controls the output actually generated by the enforcement mechanism. A \REDUCE\ program $\comm_R$ can ask an item from a local execution, send an item to the external output, clean local output queues of local executions and wake local executions up.

\begin{figure}
\centering
	\begin{alignat*}{2}
		\commR ::=&\dots & \spaceIns instructions:\\
		& |\iinputr{x}{i}{\chnl} & \NINPUTR\\
		& |\iclean{\chnl}{\PREDD} & \NCLEAN
	\end{alignat*}
\caption{\REDUCE\ instructions}
\label{fig:comm:REDUCE}
\end{figure}

\begin{definition}\label{def:execConf:REDUCE}
A configuration of the \REDUCE\ component is a set $\{\lreduce.\lpair{\lprogr}{\commR}, \lreduce.\lpair{\lmem}{m}\}$, where $\comm_R$ is the instruction to be executed, and $m$ is the memory.
\end{definition}

Except for the \NINPUT\ instruction that is replaced by the \NINPUTR\ instruction, in addition to the instructions in Fig.~\ref{fig:CommandSem} and the \NWAKE\ instruction, the program of the \REDUCE\ component may contain the instructions described in Fig.~\ref{fig:comm:REDUCE}. The execution of the \NINPUTR\ instruction reads the value from the output queue of $\Progl{i}$ and stores it into $x$. The execution of the \NCLEAN\ instruction is applied to all local execution \Progl{i} such that \PREDE{i} is true. This instruction removes the first vector $\vec{v}$ of the output queue \Oid\ of \Progl{i}, where the value of $\vec{v}[\chnl]$ is different from \NIL.

The semantics of the \NINPUTR, \NOUTPUT, \NWAKE, and \NCLEAN\ instructions is described in Fig.~\ref{fig:Sem:REDUCE}, where the function \remove{\Oid}{\chnl} removes the first vector $\vec{v}$ in \Oid\ where $\vec{v}[\chnl] \neq \NIL$.

\begin{figure}[!t]
\vspace{15pt}
\begin{center}
\ACCORCIA
\begin{mathpar}
\LabelRule{\RINPUTR}{\commR = \iinputr{x}{i}{\chnl} \\
	\ST[i].\loutput = \Oid \\
	\dequeue{\Oid}{\chnl} = (\valueM,\Oid') \\
	\valueM \neq \NIL}
{\Delta, \lreduce.\lprogr\lcomma\commR, \lreduce.\lmemr\lcomma m \Rightarrow \Delta,\lreduce.\lprogr\lcomma\textbf{skip}, \lreduce.\lmemr\lcomma \substitute{m}{x}{\valueM}} \and
\LabelRule{\ROUTPUTR}{\commR = \ioutput{e}{\chnl} \\
	\lreduce.\lmem = m \\
	\vec{v} = \constructV{m(e)}{\chnl}}
{\Delta, \lreduce.\lprogr\lcomma\commR, \loutput\lcomma O \Rightarrow \Delta,\lreduce.\lprogr\lcomma\textbf{skip}, \loutput\lcomma O.\vec{v}} \and
\LabelRule{\RWAKER}{\commR = \iwaker{\PREDD} \\
	S = \{i \in \{0,\dots,\TOP\}:\PREDE{i}\} \\
	\LECS = \bigcup_{i \in S} \{\ST[i].\litr\lcomma\signal, \ST[i].\lstate\lcomma\sS \} \\
	\LECS' = \bigcup_{i \in S} \{\ST[i].\litr\lcomma\NIL, \ST[i].\lstate\lcomma\sE \}}
{\Delta, \lreduce.\lprogr\lcomma\commR, \LECS  \Rightarrow \Delta, \lreduce.\lprogr\lcomma\textbf{skip},\LECS'} \and
\LabelRule{\RCLEAN}{\commR = \iclean{\chnl}{\PREDD} \\
	S = \{i \in \{0,\dots,\TOP\}:\PREDE{i}\} \\
	\LECS = \bigcup_{i \in S} \{\ST[i].\loutput\lcomma\Oid\} \\
	\LECS' = \bigcup_{i \in S} \{\ST[i].\loutput\lcomma\remove{\Oid}{\chnl} \}}
{\Delta, \lreduce.\lprogr\lcomma\commR, \LECS \Rightarrow \Delta,\lreduce.\lprogr\lcomma\textbf{skip},\LECS'}\and
\LabelRule{\RREDUCEINIT}{\lreduce.\lprogm\lcomma\iskip \\
	S = \{i \in \{0,\dots,\TOP\}:\INIREDUCEE{i}\} \\
	S \neq \emptyset\\\\ i = \select{S} \\
	\ST[i].\lprog = \ioutput{e}{\chnl};\comm}
{\Delta, \ST[i].\litr\lcomma\chnl,\lreduce.\lprogm\lcomma\iskip, \lreduce.\lmemr\lcomma m \\\\
 \Rightarrow
\Delta, \ST[i].\litr\lcomma\NIL, \lreduce.\lprogm\lcomma \ProgR(i,\chnl),\lreduce.\lmemm\lcomma m_0}
\end{mathpar}
\ACCORCIA
\caption{Semantics of instructions of \REDUCE}
\label{fig:Sem:REDUCE}
\end{center}
\end{figure}

The initial configuration of \REDUCE\ is $\{\lreduce.\lprogr\lcomma\NSKIP, \lreduce.\lmem\lcomma m_0 \}$. Similar to the execution of \MAP, the execution of \REDUCE\ is terminated if \NSKIP\ is the only instruction in the \REDUCE\ program.


We define the predicate \INIREDUCED\ indicating whether a local execution is sleeping on an output instruction.

\begin{eqnarray*}
\INIREDUCED & \triangleq & \lambda x. \ST[x].\lstate = \sS \wedge \exists \chnl \in \Cout:\ST[x].\litr = \intsig{\chnl} \wedge \nonumber \\
			& & {} \wedge \ST[x].\lprog = \ioutput{e}{\chnl};\comm
\end{eqnarray*}

In Fig.~\ref{fig:Sem:REDUCE} we present the \REDUCE\ activation rule ($\RREDUCEINIT$).  Similarly to \MAP, \REDUCE\ can be activated when there is only the \NSKIP\ instruction in the \REDUCE\ program, and there is an interrupt signal  \intsig{\chnl} from the local execution \Progl{i}, the state of this local execution is sleeping (\sS), the instruction to be executed is an output instruction. The activation of  \REDUCE\ on a signal on channel \chnl\ from \Progl{i} will remove the signal from configuration of \Progl{i}.

\section{Configurations for the Selected Properties}\label{sec:em}

In \cite{MANT-00-CSF}, Mantel proposes a uniform framework to define possibilistic information flow properties and he proves that existing possibilistic information flow properties can be expressed as a predefined basic security predicate (BSP) or conjunction of these BSPs. A BSP is generally defined in the framework of Mantel based on removal of some high inputs and events.



In the next sections, we will demonstrate configurations of our framework for enforcement of two BSPs, RI and DI, and the SME-style NI. It might not be obvious whether these properties are actually different in our model. We resolve possible doubts of the attentive reader in Sec.~\ref{sec:discussion:relationship}.

\begin{definition}
Let $\COND\ \triangleq \lambda \vec{v}. Cond()$ be a predicate and \CONDP{\vec{v}} be the result of the evaluation of \COND\ on $\vec{v}$. We define the restriction operator on the queue $Q$ with \COND\ as follows:
\begin{equation*}
\restrict{Q}{\COND} \triangleq \begin{cases}
	\emptyQ, &\text{if $Q = \emptyQ$};\\
	\vec{v}.\restrict{Q'}{\COND}, &\text{if $Q = \vec{v}.Q'$ and $\CONDP{\vec{v}} = \VTRUE$}; \\
	\restrict{Q'}{\COND}, &\text{if $Q = \vec{v}.Q'$ and $\CONDP{\vec{v}} = \VFALSE$}.
	\end{cases}
\end{equation*}
\end{definition}

We will use the notation \restrict{Q}{l}, the restriction on security level $l$, if $Cond(l) \triangleq \lambda \vec{v}. \exists \chnl:\vec{v}[\chnl] \neq \NIL \wedge \LVL[\chnl] = l$; and the notation \restrict{Q}{\chnl}, the restriction on channel \chnl, if $Cond(c) \triangleq \lambda \vec{v}. \vec{v}[\chnl] \neq \NIL$.


In the sequel, we will use the program described in Fig.~\ref{fig:example:source} to illustrate how the enforcement mechanism works for different information flow properties. The program has two high input channels \linecode{cH1}, \linecode{cH2}, and one high output channel \linecode{cH3}. With the execution of instructions at lines~\ref{example:HInfL1}, \ref{example:HInfL2}, \ref{example:LInfH}, and \ref{example:InputH2}, the secret values from \linecode{cH1} (line~\ref{example:InputH1}) and \linecode{cH2} (line~\ref{example:InputH2}) can influence the value sent to the low output channel \linecode{cL3} (line~\ref{example:outputL}). In addition, the sequences of high input items are effected by the low input (line~\ref{example:LInfH} and \ref{example:InputH2}); for example, if the value of \linecode{l1} is \VTRUE, an input item from \linecode{cH2} will be consumed.

We consider the execution of the program with the input sequence \qItem{cH1}{\VTRUE} \qItem{cL1}{\VFALSE} \qItem{cL2}{m} \qItem{cH2}{M}. 

\begin{figure}[!t]
\begin{javascript}
input h1 from cH1          (*\label{example:InputH1}*)
input l1 from cL1
if !h1 then                (*\label{example:HInfL1}*)
    l1 := !l1              (*\label{example:HInfL2}*)
input l2 from cL2
h2 := 0
if l1 then	               (*\label{example:LInfH}*)
    input h2 from cH2      (*\label{example:InputH2}*)
output l2 + h2 to cH3      (*\label{example:outputH}*)
output l2 + h2 to cL3	   (*\label{example:outputL}*)
\end{javascript}
\caption{Running Example Program}
\label{fig:example:source}
\end{figure}

\subsection{Removal of Inputs}\label{sec:em:RI}

The property of removal of inputs \cite{MANT-00-CSF} requires that if a possible trace is perturbed by removing all high input items, then the result can be corrected into a possible trace.
\begin{equation*} \label{equa:RI:Mantel}
\forall t \in Tr. \exists t' \in Tr. t|_L = t'|_L \wedge t'|_{HI} = \epsilon
\end{equation*}

In our notation, if all high input items in an input queue are replaced by default items or removed, the input queue can be modified to an input such that that the program will be terminated when executing on this input and the output generated will be equivalent at the low level with the original output.

\begin{definition}\label{def:RI}
A program $\Prog$ satisfies the property of \emph{removal of inputs} iff

\begin{multline*}
 \forall I, \forall\ values\ of\ \defVal: \execution{\Prog}{I}{O} \implies  \exists I': \loweq{I'}{I} \wedge I'|_H = (\defvec)^*  \wedge \\
    \wedge \forall \chnl \in \Cin, \length{I'|_c} \leq \length{I|_c} \wedge \execution{\Prog}{I'}{O'} \wedge \loweq{O'}{O},
\end{multline*}where the vector \defvec\ contains the default value, and \length{Q} returns the length of $Q$.
\end{definition}


The enforcement mechanism of the RI property on the program \Prog\ only needs two parallel programs: the high ($\Prog[0]$) and the low ($\Prog[1]$). We specify the full configuration of the local executions in Fig.~\ref{fig:table:RI}. The high execution can receive (real) input values from $L$ and $H$ channels, while the low execution can receive only (real) input values from $L$ channels. The high execution can write output values only to $H$ channels, the low execution can write values only to $L$ channels. If the interrupt signal is from \Progl{1}, or the interrupt signal is from \Progl{0} and the level of channel \chnl\ is $H$, then the input action will be performed. Otherwise, the local execution keeps sleeping.

\begin{figure}[!t]
\subfloat[\MAP\ for RI for an input from \chnl\ from \Progl{i}]{
\label{alg:MAP:RI}
\fbox{
\hspace{-12pt}
\begin{minipage}{0.45\columnwidth}
\begin{algorithmic}[1]
    \IF {$a \in \TAV[i][\chnl]$}                            \label{alg:RI:MAP:l1}
        \STATE \iinput{x}{\chnl} 							\label{alg:RI:MAP:l2}
        \STATE \imap{x}{\chnl}{\canMap{\chnl}}				\label{alg:RI:MAP:l3}
        \STATE \imap{\defVal}{\chnl}{\neg \canMap{\chnl}}	\label{alg:RI:MAP:l4}
        \STATE \iwake{\isReady{\chnl}}						\label{alg:RI:MAP:l5}
    \ELSE													\label{alg:RI:MAP:l6}
        \STATE \iskip										\label{alg:RI:MAP:l7}
    \ENDIF
\end{algorithmic}
\end{minipage}
}
}
\hfill
\subfloat[\REDUCE\ for RI for an output to \chnl\ from \Progl{i}]{
\label{alg:REDUCE:RI}
\fbox{
\hspace{-12pt}
\begin{minipage}{0.49\columnwidth}
\begin{algorithmic}[1]
    \STATE $x := \defVal$                      \label{alg:RI:REDUCE:l1}
	\IF{$\tput \in \TPV[i][\chnl]$}            \label{alg:RI:REDUCE:l2}
	    \STATE $\iinputr{x}{i}{c}$             \label{alg:RI:REDUCE:l3}
	\ENDIF
	\IF {$\ttell \in \TPV[i][\chnl]$}          \label{alg:RI:REDUCE:l4}
        \STATE $\ioutput{x}{\chnl}$            \label{alg:RI:REDUCE:l5}
    \ENDIF
    \STATE \iclean{\chnl}{\identical{i}}       \label{alg:RI:REDUCE:l6}
    \STATE \iwake{\identical{i}}               \label{alg:RI:REDUCE:l7}
\end{algorithmic}
\end{minipage}}
} \\
\subfloat[\TAV]{
\label{fig:table:RI:MAP}
\begin{tabular}{|c|c|c|}
\hline
~ & \tindex{0} & \tindex{1} \\
\hline
$\LVL[\chnl] = H$ & \tgetask & \task \\
\hline
$\LVL[\chnl] = L$ & \tget & \tgetask \\
\hline
\end{tabular}
}
\hfill
\subfloat[\TPV]{
\label{fig:table:RI:REDUCE}
\begin{tabular}{|c|c|c|}
\hline
~ &  \tindex{0} & \tindex{1} \\
\hline
$\LVL[\chnl] = H$ & \ttellput & \tnoaction \\
\hline
$\LVL[\chnl] = L$ & \tnoaction & \ttellput \\
\hline
\end{tabular}
}\\
\figdesc{RI prevents attackers from inferring what high input items have been read, since the removal and the replacement of high input items with fake ones have no effect on what is visible to attackers. Only the low execution in the enforcement mechanism can output to low channels. \MAP\ can perform input actions for all requests from the low execution. The low execution can only receive default values for high input items.}
\caption{Configuration for the enforcement mechanism of RI}
\label{fig:table:RI}
\end{figure}

The program of \MAP\ is described in Fig.~\ref{alg:MAP:RI}. The function $\canMap{c}$ indicates whether the local execution \Progl{x} can receive real values from \MAP:
\begin{equation}  \label{func:canMap}
\canMap{\chnl} \triangleq \lambda x. \ttell \in \TAV[x][\chnl]
\end{equation}

If a local execution that is sleeping and waiting for an input item from a channel has received the input item required, this local execution is ready to be waken up:
\begin{eqnarray}\label{func:isReady}
\isReady{c} & \triangleq & \lambda x. \ST[x].\lstate = \sS \wedge \ST[x].\lprog = \textbf{input}~ y ~\textbf{from}~ c;\comm \wedge \nonumber \\
			& & \wedge \ST[x].\linput = \Iid \wedge \dequeue{\Iid}{\chnl} = (\valueM, I') \wedge  \valueM \neq \NIL
\end{eqnarray}


When there is an interrupt signal on channel \intsig{\chnl} from \Progl{i} on an output instruction, the program \REDUCE\ provided in Fig.~\ref{alg:REDUCE:RI} is activated. If the local execution $\Progl{i}$ can send items to channel \chnl\ ($\TPV[i][\chnl] = 1$), the output action is performed. Otherwise, there is no output action. After that, the output queue of $\Prog[i]$ is cleaned and only $\Progl{i}$ is waken. Since the execution of the wake instruction wakes only \Progl{i} up, the function $\identical{}$ is defined as
\begin{equation} \label{func:identical}
	\identical{i} \triangleq \lambda x. x = i
\end{equation}

\paragraph{Example.}
We consider the tables \TAV\ and \TPV\ of the enforcement mechanism of RI for the program described in Fig.~\ref{fig:example:source}. In Figure~\ref{fig:table:RI:MAP:example} we show the \TAV\ instantiation for the given set of channels. Following Fig.~\ref{fig:table:RI:MAP}, on the channels \linecode{cH1} and \linecode{cH2}  we assign the ``$\task\ttell$'' permissions for the high execution \Progl{0}, but only the ``$\task$'' permission to the low execution \Progl{1}.  We also assign the permission ``$\ttell$'' on the low channels \linecode{cL1} and \linecode{cL2} for the high execution \Progl{0}, while we set ``$\task\ttell$'' for the low execution \Progl{1}. The table \TPV\ in Fig.~\ref{fig:table:RI:REDUCE:example} is configured similarly, following Fig.~\ref{fig:table:RI:REDUCE}.

\begin{figure}
\centering
\subfloat[\TAV]{
\label{fig:table:RI:MAP:example}
\begin{minipage}{0.30\columnwidth}
\centering
\begin{tabular}{|c|c|c|}
\hline
~ & \tindex{0} & \tindex{1} \\
\hline
\linecode{cH1} & \tgetask & \task \\
\hline
\linecode{cH2} & \tgetask & \task \\
\hline
\linecode{cL1} & \tget & \tgetask \\
\hline
\linecode{cL2} & \tget & \tgetask \\
\hline
\end{tabular}
\end{minipage}
} 
\subfloat[\TPV]{
\label{fig:table:RI:REDUCE:example}
\begin{minipage}{0.30\columnwidth}
\centering
\vspace{25pt}
\begin{tabular}{|c|c|c|}
\hline
~ &  \tindex{0} & \tindex{1} \\
\hline
\linecode{cH3} & \ttellput & \tnoaction \\
\hline
\linecode{cL3} & \tnoaction & \ttellput \\
\hline
\end{tabular}
\end{minipage}
}
\caption{Example tables \TAV\ and \TPV\ for RI}
\label{fig:example:execution:RI:TRTR}
\end{figure}

Examples of executions of the high execution \Progl{0} and the low execution \Progl{1} are shown in Fig.~\ref{fig:example:execution:RI:HLEx}. Here, we assume that the high execution runs faster. At line~\ref{example:RI:HEx:l1}, \Progl{0} sends a request to \MAP\ and moves to the state \sS. \MAP\ is activated and the code from line~\ref{alg:RI:MAP:l1} to line~\ref{alg:RI:MAP:l5} in Fig.~\ref{alg:MAP:RI} is executed, since the permission ``$\task$'' is in $\TAV[0][\linecode{cH1}]$. Line~\ref{alg:RI:MAP:l2} reads the value \VTRUE\ from the global queue; line~\ref{alg:RI:MAP:l3} sends \VTRUE\ to the local input queue of \Progl{0}, since \canMap{\chnl} returns \VTRUE\ with \Progl{0}; line~\ref{alg:RI:MAP:l4} sends the default value \VFALSE\ to the local input queue of \Progl{1}; line~\ref{alg:RI:MAP:l5} wakes \Progl{0} up, since \isReady{\linecode{cH1}} returns \VTRUE\ with \Progl{0}. The high execution \Progl{0} is waken up and it continues to execute at line~\ref{example:RI:HEx:l1}. At this time, since there is a value from \linecode{cH1} in the local input queue, the execution of the input instruction follows the rule \RINPUTYL. Next, line~\ref{example:RI:HEx:l2} in Fig.~\ref{fig:example:RI:execution:H} is executed. The high execution moves to the state \sS\ and waits for an input item from \linecode{cL1}.

\begin{figure}[!t]
\centering
\input{RI-Example-H}
\subfloat[The high execution \Progl{0}]{\label{fig:example:RI:execution:H}\usebox{\mylistingbox}}%
\input{RI-Example-L}
\hfill
\subfloat[The low execution \Progl{1}]{\label{fig:example:RI:execution:L}\usebox{\mylistingbox}}%
\caption{Executions of local copies for RI}
\label{fig:example:execution:RI:HLEx}
\end{figure}

The low execution starts executing. The execution of line~\ref{example:RI:LEx:l1} follows the rule \RINPUTYL\ since there is an item (with default value) from \linecode{cH1} in the local input queue. The execution of line~\ref{example:RI:LEx:l2} moves the low execution to the state \sS. \MAP\ is activated, and it consumes \VFALSE\ from \linecode{cL1} (line~\ref{alg:RI:MAP:l2}), sends \VFALSE\ to both local input queues (line~\ref{alg:RI:MAP:l3}), and wakes the low execution \Progl{1} up (line~\ref{alg:RI:MAP:l5}). The execution of line~\ref{alg:RI:MAP:l4} does nothing since there is no local execution that can make $\neg\canMap{cL1}$ be true. After that, the low execution keeps executing.

Because values of \linecode{h1} and \linecode{l1} are respectively \VTRUE\ and \VFALSE, the high execution executes instructions at lines~\ref{example:RI:HEx:l3}, \ref{example:RI:HEx:l5}, \ref{example:RI:HEx:l6}, \ref{example:RI:HEx:l7}, \ref{example:RI:HEx:l9}, and \ref{example:RI:HEx:l10}. At line~\ref{example:RI:HEx:l9} the high execution moves to the state \sS\ and then \REDUCE\ is activated. The program described in Fig.~\ref{alg:REDUCE:RI} is executed. Since the permission ``$\task\ttell$'' is in \TPV[0][\linecode{cH3}], \REDUCE\ retrieves the value generated by the high execution (line~\ref{alg:RI:REDUCE:l3}), sends the value to the global output queue (line~\ref{alg:RI:REDUCE:l5}), cleans the output queue of \Progl{1} (line~\ref{alg:RI:REDUCE:l6}), and wakes the high execution \Progl{0} up (line~\ref{alg:RI:REDUCE:l7}).

\begin{figure}[!t]

\begin{lrbox}{\mylistingbox}%
\begin{tabular}{l}
\hspace{3pt}
\begin{minipage}{183pt}
The input to \MAP:\\
\begin{tabular}{|c|c|c|c|c|}
	\hline
	\backslashbox{Channel}{Time} & 0 & 1 & 2 & 3  \\
    \hline
	\linecode{cH1} & $\VTRUE$ & \NIL & \NIL & \NIL\\
	\hline
	\linecode{cH2} & \NIL & \NIL & \NIL & $M$ \\
	\hline
	\linecode{cL1} & \NIL & \VFALSE & \NIL & \NIL\\
	\hline
	\linecode{cL2} & \NIL & \NIL & $m$ & \NIL\\
	\hline 
\end{tabular}  
\end{minipage} $\Longrightarrow$ \MAP \\

\begin{tabular}{l}
Local Executions: \\
\begin{tabular}{|l|}
\hline
The high execution \Progl{0}:\\
\begin{tabular}{p{110pt}p{180pt}}
The local input: & The local output:\\
\begin{tabular}{|c|c|c|c|c|}
	\hline
	\linecode{cH1} & $\VTRUE$ & \NIL & \NIL & \NIL\\
	\hline
	\linecode{cH2} & \NIL & \NIL & \NIL & $M$\\
	\hline
	\linecode{cL1} & \NIL & \VFALSE & \NIL & \NIL\\
	\hline
	\linecode{cL2} & \NIL & \NIL & $m$ & \NIL\\
	\hline
\end{tabular} & 
\begin{minipage}{105pt}
\begin{tabular}{|c|c|c|c|c|>{\centering\arraybackslash}p{25pt}|>{\centering\arraybackslash}p{25pt}|}
	\hline	
	\linecode{cH3} & \NIL & \NIL & \NIL & \NIL & $m$ & \NIL\\
	\hline
	\linecode{cL3} & \NIL & \NIL & \NIL & \NIL & \NIL & $m$\\
	\hline 	
\end{tabular}
\vspace{25pt}
\end{minipage} \\

\end{tabular} \\\\
\hline
The low execution \Progl{1}:\\
\begin{tabular}{p{110pt}p{180pt}}
The local input: & The local output:\\
\begin{tabular}{|c|c|c|c|c|}
	\hline
	\linecode{cH1} & $\VFALSE$ & \NIL & \NIL & \NIL\\
	\hline
	\linecode{cH2} & \NIL & \NIL & \NIL & $*$\\
	\hline
	\linecode{cL1} & \NIL & \VFALSE & \NIL & \NIL\\
	\hline
	\linecode{cL2} & \NIL & \NIL & $m$ & \NIL\\
	\hline
\end{tabular} &
\begin{minipage}{180pt}
\begin{tabular}{|c|c|c|c|c|>{\centering\arraybackslash}p{25pt}|>{\centering\arraybackslash}p{25pt}|}
	\hline	
	\linecode{cH3} & \NIL & \NIL & \NIL & \NIL & $* + m$ & \NIL\\
	\hline
	\linecode{cL3} & \NIL & \NIL & \NIL & \NIL & \NIL & $* + m$\\
	\hline
\end{tabular}
\vspace{24pt}
\end{minipage}
\end{tabular} \\\\ \hline
\end{tabular} 
\end{tabular} \\
\vspace{3pt}
\hspace{16pt}
\REDUCE\ $\Longrightarrow$ 
\begin{minipage}{180pt}
~\\
The output by \REDUCE:\\
\begin{tabular}{|c|c|c|c|c|>{\centering\arraybackslash}p{25pt}|>{\centering\arraybackslash}p{25pt}|}
	\hline
	\backslashbox{Channel}{Time} & 0 & 1 & 2 & 3 & 4 & 5  \\
	\hline	
	\linecode{cH3} & \NIL & \NIL & \NIL & \NIL & $m$ & \NIL\\
	\hline
	\linecode{cL3} & \NIL & \NIL & \NIL & \NIL & \NIL & $* + m$\\
	\hline
\end{tabular}
\end{minipage}
\end{tabular}
\end{lrbox}


\usebox{\mylistingbox}
\caption{Example of input and output queues for RI}
\label{fig:example:RI:execution:io}
\end{figure}

Since the value of both \linecode{h1} and \linecode{l1} is \VFALSE, the instructions at lines~\ref{example:RI:LEx:l3}, \ref{example:RI:LEx:l4}, \ref{example:RI:LEx:l5}, \ref{example:RI:LEx:l6}, \ref{example:RI:LEx:l7}, \ref{example:RI:LEx:l8}, \ref{example:RI:LEx:l9}, and \ref{example:RI:LEx:l10} of the low execution \Progl{1} are executed. At line~\ref{example:RI:LEx:l8} the low execution moves to the state \sS\ and \MAP\ is activated. Because the permission ``$\task$'' is in $\TAV[1][\linecode{cH2}]$, \MAP\ reads $M$ from the global input queue, sends $M$ and a default value ($*$) to the local input queues of the high and the low execution respectively, and wakes the low execution up. When executing the output instruction at line~\ref{example:RI:LEx:l9} the low execution moves to the state \sS, and \REDUCE\ is activated. Since neither of the permissions ``$\task$'' and ``$\ttell$'' is in $\TPV[1][\linecode{cH3}]$, \REDUCE\ does not retrieve any item from the local output queue of the low execution and does not send any value to the global queue. \REDUCE\ cleans the output generated by \Progl{1} and wakes \Progl{1} up. The execution of the output instruction at line~\ref{example:RI:LEx:l10} is similar to the execution of the output instruction at line~\ref{example:RI:HEx:l9} of the high execution.

We describe the global input, output queues, and local input, output queues in Fig.~\ref{fig:example:RI:execution:io}. The global input queue is consumed completely by the execution of the enforcement mechanism. The values sent to \linecode{cH3} and \linecode{cL3} are respectively $m$ and $* + m$. Each column in the table corresponds to an input/output operation. Input and output tables should be read from left to right; columns describe the input/output to each channel at time $t = 0$, $t = 1$, etc.

\subsection{Non-Interference}\label{sec:em:NI}
The enforcement mechanism configured in this section mimics the SME-style enforcement of non-interference \cite{Devr-Pies-10-IEEESP} from Devriese and Piessens, and therefore inherits also the limitations of SME formal guarantees.

\begin{figure}[!t]
\subfloat[\MAP\ for SME for input from \chnl\ requested from \Progl{i}]{
\label{alg:MAP:NI:SME}
\fbox{
\hspace{-12pt}
\begin{minipage}{0.47\columnwidth}
\begin{algorithmic}[1]

    \IF {$\task \in \TAV[i][\chnl]$}                        \label{alg:NI:MAP:l1}
        \STATE \iinput{x}{\chnl}                     	    \label{alg:NI:MAP:l3}
        \STATE \imap{x}{\chnl}{\NCMAP(c)}               	\label{alg:NI:MAP:l4}
        \STATE \imap{\defVal}{\chnl}{\neg\NCMAP(c)}     	\label{alg:NI:MAP:l5}
        \STATE \iwake{\isReady{\chnl}}                  	\label{alg:NI:MAP:l6}
    \ELSE                                               	\label{alg:NI:MAP:l7}
    	\IF {$\ttell \not\in \TAV[i][\chnl]$}				\label{alg:NI:MAP:l2}
            \STATE \imap{\defVal}{\chnl}{\identical{i}}     \label{alg:NI:MAP:l8}
            \STATE \iwake{\identical{i}}                    \label{alg:NI:MAP:l9}
        \ELSE                                               \label{alg:NI:MAP:l10}
	        \STATE \iskip                                   \label{alg:NI:MAP:l11}
	    \ENDIF
    \ENDIF

\end{algorithmic}
\end{minipage}
}
}
\hfill
\subfloat[\REDUCE\ for SME for an output to \chnl\ from \Progl{i}]{
\label{alg:REDUCE:NI}
\begin{minipage}{0.47\columnwidth}
\vspace{45pt}
\fbox{
\hspace{-12pt}
\begin{minipage}{\columnwidth}
\begin{algorithmic}[1]
    \STATE $x := \defVal$
	\IF{$\tput \in \TPV[i][\chnl]$}
	    \STATE $\iinputr{x}{i}{c}$
	\ENDIF
	\IF {$\ttell \in \TPV[i][\chnl]$}
        \STATE $\ioutput{x}{\chnl}$
    \ENDIF
    \STATE \iclean{\chnl}{\identical{i}}
    \STATE \iwake{\identical{i}}
\end{algorithmic}
\end{minipage}}
\end{minipage}
} \\
\subfloat[\TAV]{
\label{fig:table:NI:MAP}
\begin{tabular}{|c|c|c|}
\hline
~ & \tindex{0} & \tindex{1} \\
\hline
$\LVL[\chnl] = H$ & \tgetask & \tnoaction \\
\hline
$\LVL[\chnl] = L$ & \tget & \tgetask \\
\hline
\end{tabular}
}
\hfill
\subfloat[\TPV]{
\label{fig:table:NI:REDUCE}
\begin{tabular}{|c|c|c|}
\hline
~ &  \tindex{0} & \tindex{1} \\
\hline
$\LVL[\chnl] = H$ & \ttellput & \tnoaction \\
\hline
$\LVL[\chnl] = L$ & \tnoaction & \ttellput \\
\hline
\end{tabular}
}\\
\figdesc{NI prevents attackers from inferring high input items, since all executions on inputs that are low-equivalent will generate low-equivalent outputs. The low execution \Progl{1} cannot ask values from high input channels or be told  real value from these channels. \MAP\ will not perform any input actions on requests for high inputs items from the low execution. \TPV\ and the program of \REDUCE\ are the same as in the enforcement mechanism of RI.}
\caption{Configuration for the enforcement mechanism of NI}
\label{fig:table:NI}
\end{figure}

Informally, a program satisfies the termination-insensitive non-interference (TINI) property if given two arbitrary inputs that are equivalent at the low level and the executions of the program on these two inputs are terminated, then the outputs generated are indistinguishable to the users at the low level. In other words, the high input items in these two inputs have no effect on  what observable is to users at the low level. Termination-sensitive non-interference (TSNI) additionally requires that the secret input items do not influence the termination of the program \cite{Bart-Arge-Rezk-11-MSCS}.

\begin{definition}\label{def:NI}
A program $\Prog$ satisfies the property of \emph{termination-insensitive non-interference} iff
\begin{equation*}
\forall I, I': \loweq{I'}{I} \implies \loweq{O'}{O},
\end{equation*}
where \execution{\Prog}{I}{O} and \execution{\Prog}{I'}{O'}.
\end{definition}

\begin{definition}\label{def:TSNI}
A program $\Prog$ satisfies the property of \emph{termination-sensitive non-interference} iff
\begin{equation*}
 \forall I, I': \loweq{I'}{I} \wedge \execution{\Prog}{I}{O} \implies \execution{\Prog}{I'}{O'} \wedge \loweq{O'}{O}
\end{equation*}
\end{definition}

To implement the SME approach \cite{Devr-Pies-10-IEEESP}, we use the following configuration. The high execution $\Progl{0}$ can only read high input items, but it can consume both high and low ones. For low input items this local execution needs to wait for the values read by the low execution \Progl{1}. The low execution $\Progl{1}$ can read and consume only low items. If the low execution requires a high input item, the default value will be used.

The configuration tables \TAV\ and \TPV\ and the program for \MAP, and \REDUCE\ to enforce the SME-style NI are presented in Fig.~\ref{fig:table:NI}. Compared to the program for \MAP\ in the enforcement mechanism of RI, the program for \MAP\ is different, as shown in Fig.~\ref{alg:MAP:NI:SME}. The functions $\canMap{\chnl}$, $\isReady{\chnl}$ and $\identical{i}$ are defined in respectively Eq.~\ref{func:canMap},~\ref{func:isReady} and \ref{func:identical} in~\S.~\ref{sec:em:RI}.

\paragraph{Example.}
The contents of the tables \TAV\ and \TPV\ for the enforcement mechanism of NI for the running example program in Fig.~\ref{fig:example:source} follow the patterns described in Fig.~\ref{fig:table:NI:MAP} and Fig.~\ref{fig:table:NI:REDUCE} respectively. 

\begin{figure}[!t]
\input{NI-Example-H}
\subfloat[The high execution \Progl{0}]{\label{fig:example:NI:execution:H}\usebox{\mylistingbox}}%
\input{NI-Example-L}\\
\subfloat[The low execution \Progl{1}]{\label{fig:example:NI:execution:L}\usebox{\mylistingbox}}%
\\
\figdesc{When the low execution \Progl{1} asks \MAP\ an item from \linecode{cH2}, \MAP\ does not perform any input operation and sends a default value to the local input queue of the low execution}
\caption{Executions of local copies for NI}
\label{fig:example:execution:NI:HLUEx}
\end{figure}

The executions of the high execution and the low execution are described in Fig.~\ref{fig:example:execution:NI:HLUEx}. The execution of the high program is similar to the execution of the high program in RI. The execution of the low execution is almost the same as the execution of the low execution in RI, except for handling the input instruction at line~\ref{example:NI:LEx:l8}.

When the input instruction at line~\ref{example:NI:LEx:l8} is executed, the low execution moves to the state \sS\ and sends a signal to \MAP. When \MAP\ is activated on this signal, the instructions at lines~\ref{alg:NI:MAP:l1}, \ref{alg:NI:MAP:l2}, \ref{alg:NI:MAP:l8}, \ref{alg:DI:MAP:l9} in the program described in Fig.~\ref{alg:MAP:NI:SME} are executed. \MAP\ sends a default value to the local input queue of the low execution (line~\ref{alg:DI:MAP:l8}), and wakes this local execution up (line~\ref{alg:DI:MAP:l9}).

We describe the global input, output queues, and local input, output queues in Fig.~\ref{fig:example:execution:NI:IO}. Compared to the execution of the enforcement mechanism of RI, the execution of the enforcement of NI consumes only three input items. The fourth input item is kept in the figure for the illustrative purposes and in order to make a comparison with the enforcement of RI. The local input queue of the high execution contains only three input items; these are two low input items requested by the low execution and one high input item requested by the high execution. The local input queue of the low execution \Progl{1} contains four input items, where the last one is the default item generated and sent by \MAP.

\begin{figure}[!t]
\begin{lrbox}{\mylistingbox}%
\begin{tabular}{l}
\begin{minipage}{183pt}
The input to \MAP:\\
\emph{(The last input item is not consumed.)} 
\begin{tabular}{|c|c|c|c|c|}
	\hline
	\backslashbox{Channel}{Time} & 0 & 1 & 2 & 3 \\
    \hline
	\linecode{cH1} & $\VTRUE$ & \NIL & \NIL & \NIL\\
	\hline
	\linecode{cH2} & \NIL & \NIL & \NIL & $M$ \\
	\hline
	\linecode{cL1} & \NIL & \VFALSE & \NIL & \NIL\\
	\hline
	\linecode{cL2} & \NIL & \NIL & $m$ & \NIL\\
	\hline 
\end{tabular} \\
\end{minipage} $\Longrightarrow$ \MAP \\
Local Executions: \\
\begin{tabular}{|l|}
\hline
The high execution \Progl{0}:\\
\begin{tabular}{p{110pt}p{180pt}}
The local input: & The local output:\\
\begin{tabular}{|c|c|c|c|c|}
	\hline
	\linecode{cH1} & $\VTRUE$ & \NIL & \NIL \\
	\hline
	\linecode{cH2} & \NIL & \NIL & \NIL \\
	\hline
	\linecode{cL1} & \NIL & \VFALSE & \NIL \\
	\hline
	\linecode{cL2} & \NIL & \NIL & $m$ \\
	\hline
\end{tabular} & 
\begin{minipage}{180pt}
\begin{tabular}{|c|c|c|c|c|>{\centering\arraybackslash}p{25pt}|>{\centering\arraybackslash}p{25pt}|}
	\hline	
	\linecode{cH3} & \NIL & \NIL & \NIL & \NIL & $m$ & \NIL\\
	\hline
	\linecode{cL3} & \NIL & \NIL & \NIL & \NIL &  \NIL & $m$\\
	\hline 	
\end{tabular}
\vspace{25pt}
\end{minipage} \\

\end{tabular} \\\\
\hline
The low execution \Progl{1}:\\
\begin{tabular}{p{110pt}p{105pt}}
The local input: & The local output:\\
\begin{tabular}{|c|c|c|c|c|}
	\hline
	\linecode{cH1} & $\VFALSE$ & \NIL & \NIL & \NIL\\
	\hline
	\linecode{cH2} & \NIL & \NIL & \NIL & $*$\\
	\hline
	\linecode{cL1} & \NIL & \VFALSE & \NIL & \NIL\\
	\hline
	\linecode{cL2} & \NIL & \NIL & $m$ & \NIL\\
	\hline
\end{tabular} &
\begin{minipage}{180pt}
\begin{tabular}{|c|c|c|c|c|>{\centering\arraybackslash}p{25pt}|>{\centering\arraybackslash}p{25pt}|}
	\hline	
	\linecode{cH3} & \NIL & \NIL & \NIL & \NIL & $* + m$ & \NIL\\
	\hline
	\linecode{cL3} & \NIL & \NIL & \NIL & \NIL & \NIL & $* + m$\\
	\hline
\end{tabular}
\vspace{24pt}
\end{minipage}
\end{tabular} \\\\ \hline

\end{tabular}  \\
\vspace{3pt}
\hspace{9pt}
\REDUCE\ $\Longrightarrow$ 
\begin{minipage}{180pt}
~\\
The output by \REDUCE:\\
\begin{tabular}{|c|c|c|c|c|>{\centering\arraybackslash}p{25pt}|>{\centering\arraybackslash}p{25pt}|}
	\hline
	\backslashbox{Channel}{Time} & 0 & 1 & 2 & 3 & 4 & 5 \\
	\hline	
	\linecode{cH3} & \NIL & \NIL & \NIL & \NIL & $m$ & \NIL\\
	\hline
	\linecode{cL3} & \NIL & \NIL & \NIL & \NIL & \NIL & $* + m$\\
	\hline
\end{tabular}
\end{minipage}
\end{tabular}
\end{lrbox}
\usebox{\mylistingbox}
\caption{Example of input and output queues for NI}
\label{fig:example:execution:NI:IO}
\end{figure}


\subsection{Deletion of Inputs} \label{sec:em:DI}
The property of deletion of inputs (DI) \cite{MANT-00-CSF} requires that if we perturb a possible trace $t$ (where $t = \beta.e.\alpha$ and there is no high input event in $\alpha$) by deleting the high input event $e$, then the result can be corrected into a possible trace $t'$ ($t' = \beta'.\alpha'$). The parts $\beta$ and $\beta'$ are equivalent on the low input events and the high input events. In other words, the low input events and the high input events in $\beta$ and $\beta'$ must be the same. The parts $\alpha$ and $\alpha'$ are also equivalent on the low events and the high input events. Since there is no high input events in $\alpha$, there is also no high input events in $\alpha'$.

\begin{multline*}\label{equa:DI:Mantel}
\forall t \in Tr, \forall \alpha, \beta \in E^*, \forall e \in E. (e \in HI \wedge t =\beta.e.\alpha \wedge\ \alpha|_{HI} = \epsilon) \implies \\ (\exists t' \in Tr. t|_L = t'|_L  \wedge\ t' = \beta'.\alpha'  \wedge \alpha'|_{L \cup HI} = \alpha|_{L \cup HI} \wedge \beta'|_{L \cup HI} = \beta|_{L \cup HI})
\end{multline*}

\begin{figure}[!t]
\centering
\subfloat[\MAP\ for DI for an input from \chnl\ from \Progl{i}]{
\label{alg:MAP:DI}
\fbox{
\hspace{-12pt}
\begin{minipage}{0.58\columnwidth}
\begin{algorithmic}[1]
    \IF{$\LVL[\chnl] == H$ \AND $i == 0$}                \label{alg:DI:MAP:l1}
        \STATE \iclone{\identical{i}}{\tcolm}{\tcolr}    \label{alg:DI:MAP:l2}
    \ENDIF
    \IF {$\task \in \TAV[i][\chnl]$}					 \label{alg:DI:MAP:l3}
        \IF{$\tget \in \TAV[i][\chnl]$}					 \label{alg:DI:MAP:l4}
	        \STATE \iinput{x}{\chnl}					 \label{alg:DI:MAP:l5}
	        \STATE \imap{x}{\chnl}{\NCMAP(c)}			 \label{alg:DI:MAP:l6}
	        \STATE \imap{\defVal}{\chnl}{\neg\NCMAP(c)}	 \label{alg:DI:MAP:l7}
            \STATE \iwake{\isReady{\chnl}}				 \label{alg:DI:MAP:l8}
        \ELSE											 \label{alg:DI:MAP:l9}
            \STATE \imap{\defVal}{\chnl}{\identical{i}}  \label{alg:DI:MAP:l10}
            \STATE \iwake{\identical{i}}				 \label{alg:DI:MAP:l11}
        \ENDIF	
    \ELSE												 \label{alg:DI:MAP:l12}
        \STATE \iskip									 \label{alg:DI:MAP:l13}
    \ENDIF
\end{algorithmic}
\end{minipage}
}
} \hfill
\subfloat[\REDUCE\ for DI for an output to \chnl\ from \Progl{i}]{
\label{alg:REDUCE:DI}
\begin{minipage}{0.37\columnwidth}
\vspace{70pt}
\fbox{
\hspace{-12pt}
\begin{minipage}{1\columnwidth}
\begin{algorithmic}[1]
    \STATE $x := \defVal$              \label{alg:DI:REDUCE:l1}
	\IF{$\tput \in \TPV[i][\chnl]$}    \label{alg:DI:REDUCE:l2}
	    \STATE $\iinputr{x}{i}{c}$     \label{alg:DI:REDUCE:l3}
	\ENDIF
	\IF {$\ttell \in \TPV[i][\chnl]$}  \label{alg:DI:REDUCE:l4}
        \STATE $\ioutput{x}{\chnl}$    \label{alg:DI:REDUCE:l5}
    \ENDIF
    \STATE \iclean{\chnl}{\identical{i}} \label{alg:DI:REDUCE:l6}
    \STATE \iwake{\identical{i}}         \label{alg:DI:REDUCE:l7}
\end{algorithmic}
\end{minipage}}
\end{minipage}} \\
\subfloat[\TAV]{
\label{fig:table:DI:MAP}
\begin{tabular}{|c|>{\centering\arraybackslash}p{13pt}|>{\centering\arraybackslash}p{13pt}|>{\centering\arraybackslash}p{34pt}|}
\hline
~ & \tindex{0} & \tindex{1} & $\tindex{i} > 1$\\
\hline
$\LVL[\chnl] = H$ & \tgetask & \task & \task \\
\hline
$\LVL[\chnl] = L$ & \tget & \tgetask & \tget\\
\hline
\end{tabular}
} \hfill
\subfloat[\TPV]{
\label{fig:table:DI:REDUCE}
\begin{tabular}{|c|>{\centering\arraybackslash}p{13pt}|>{\centering\arraybackslash}p{13pt}|>{\centering\arraybackslash}p{34pt}|}
\hline
~ &  \tindex{0} & \tindex{1} & $\tindex{i} > 1$\\
\hline
$\LVL[\chnl] = H$ & \ttellput & \tnoaction & \tnoaction \\
\hline
$\LVL[\chnl] = L$ & \tnoaction & \ttellput & \tnoaction \\
\hline
\end{tabular}
}\\
\figdesc{DI prevents attackers from deducing the value of the last high input item since the mechanism makes sure that if the last high input item is replaced with a fake item, the observation of the attackers is still the same. Whenever the high execution requests a high input item, the high execution will be cloned and the clone can only receive default values for high input items. When receiving a request from the low execution for a high input item, \MAP\ will return the default value. The program of \REDUCE\ is the same as in the enforcement mechanism of RI.}
\caption{Configuration for the enforcement mechanism of DI}
\label{fig:table:DI}
\end{figure}

In our notation, if we have an input queue $I$ = $I_1.\vec{v}.I_2$, where $\vec{v}$ contains a value from a high channel and in $I_2$ there are either no high input items or only high input items with default values, then this input queue can be changed by replacing $\vec{v}$ by the default vector. The obtained input queue can be sanitized by removing existing default high input items in $I_2$ or adding other default high input items to $I_2$. The sanitized queue can be consumed completely by a clone of the original program and the output should still be equivalent at the low level to the original output generated with the input $I$.

\begin{definition}\label{def:DI}
A program \Prog ~satisfies the property of \emph{deletion of inputs} $DI$ iff
\begin{multline*}
	\forall I, \forall\ values\ of\ \defVal: I = I_1.\vec{v}.I_2 \wedge LVL[\chnl] = H \wedge \defseq{I_2|_H} \wedge \execution{\Prog}{I}{O} \\ \implies
	\exists I': I' = I_1'.I_2' \wedge \loweq{I'}{I} \wedge \defseq{I_2'|_H} \wedge \execution{\Prog}{I'}{O'} \wedge \loweq{O'}{O},
\end{multline*} where $\vec{v}[\chnl] \neq \NIL$ and the vector \defvec\ contains a default value.
\end{definition}

DI is enforced with the idea that whenever the high execution requests a high input item, this execution will be cloned and the clone will not receive real values from high channels. Components configured to implement the enforcement mechanism of DI are presented in Fig.~\ref{fig:table:DI}. The program of \REDUCE\ is identical to the program used in the enforcement mechanism of RI.

The enforcement mechanism of DI requires more than two local executions. Only the high execution \Progl{0} can ask for and get the high input items, other local executions will only use default values. Each time the high execution is cloned the new execution is inserted into the stack of local executions. The configuration of the clones for input (respectively, output) is presented in Fig.~\ref{fig:table:DI:MAP} (respectively, \ref{fig:table:DI:REDUCE}) in the column $\tindex{i} > 1$; this is the privilege configuration template $\tcolm$ ($\tcolr$, respectively). In addition, only the low execution \Progl{1} can ask for low input items and generate low output items; other local executions will reuse the low input items retrieved by the low execution.

\paragraph{Example.}
Fig.~\ref{fig:example:execution:DI:TRTR} depicts \TAV\ and \TPV\  of the DI enforcement mechanism for the program described in Fig.~\ref{fig:example:source} that are instances of the tables described in Fig.~\ref{fig:table:NI:MAP} and Fig.~\ref{fig:table:NI:REDUCE} respectively. When the enforcement mechanism starts executing, there are only two columns (\Progl{0} and \Progl{1}) in each table. The third column is appended when the corresponding local execution is created. Following Fig.~\ref{fig:table:DI:MAP}, the new execution has the ``$\task$'' permission on high input channels (\linecode{cH1} and \linecode{cH2}) and the ``$\ttell$'' permission on low input channels (\linecode{cL1} and \linecode{cL2}). The permissions of the new local execution on output channels are configured in a similar way from the column $\tindex{i} > 1$ in Fig.~\ref{fig:table:DI:REDUCE}.

\begin{figure}
\centering
\subfloat[\TAV]{
\label{fig:table:DI:MAP:example}
\begin{tabular}{|c|c|c|c|}
\hline
~ & \tindex{0} & \tindex{1} & \tindex{2}\\
\hline
\linecode{cH1} & \tgetask & \task & \task\\
\hline
\linecode{cH2} & \tgetask & \task & \task\\
\hline
\linecode{cL1} & \tget & \tgetask & \tget\\
\hline
\linecode{cL2} & \tget & \tgetask & \tget\\
\hline
\end{tabular}
}
\subfloat[\TPV]{
\label{fig:table:DI:REDUCE:example}
\begin{minipage}{0.25\columnwidth}
\vspace{25pt}
\begin{tabular}{|c|c|c|c|}
\hline
~ &  \tindex{0} & \tindex{1} & \tindex{2}\\
\hline
\linecode{cH3} & \ttellput & \tnoaction & \tnoaction\\
\hline
\linecode{cL3} & \tnoaction & \ttellput & \tnoaction\\
\hline
\end{tabular}
\end{minipage}}
\caption{Example tables \TAV\ and \TPV\ for DI}
\label{fig:example:execution:DI:TRTR}
\end{figure}

The executions of local executions are described in Fig.~\ref{fig:example:execution:DI:HLUEx}. When line~\ref{example:DI:HEx:l1} of the high program is executed, the high execution moves to the state \sS; the program of \MAP\ in Fig.~\ref{alg:MAP:DI} is activated (the contents of \tcolm\ and \tcolr\ used at line~\ref{alg:DI:MAP:l2} are from the columns $\tindex{1} > 1$ in Fig.~\ref{fig:table:DI:MAP} and Fig.~\ref{fig:table:DI:REDUCE} respectively). Line~\ref{alg:DI:MAP:l2} creates a new clone of \Progl{0}, which is referred to as \Progl{2}. After that, \MAP\ consumes \VTRUE\ from the global input queue (line~\ref{alg:DI:MAP:l5}), sends \VTRUE\ to the local input queue of the high execution \Progl{0} (line~\ref{alg:DI:MAP:l6}), broadcasts \VFALSE\ to the local input queues of  \Progl{1} and \Progl{2} (line~\ref{alg:DI:MAP:l7}), and wakes \Progl{0} and \Progl{2} up (line~\ref{alg:DI:MAP:l8}).

\begin{figure}[!t]
\input{DI-Example-H}
\subfloat[The high execution \Progl{0}]{\label{fig:example:DI:execution:H}\usebox{\mylistingbox}}%
\input{DI-Example-L}\\
\subfloat[The low execution \Progl{1}]{\label{fig:example:DI:execution:L}\usebox{\mylistingbox}}%
\input{DI-Example-U} \\
\subfloat[The other execution \Progl{2}]{\label{fig:example:DI:execution:U}\usebox{\mylistingbox}}
\\
\figdesc{The execution \Progl{2} is created when the high execution requests the first high input item. When the low execution \Progl{1} (or the other execution \Progl{2}) asks \MAP\ for an item from \linecode{cH2}, \MAP\ does not perform any input operation, but it sends a default value to the local input queue of the low execution (respectively, the other local execution). }
\caption{Executions of local copies for DI}
\label{fig:example:execution:DI:HLUEx}
\end{figure}

When line~\ref{example:DI:HEx:l2} of \Progl{0} is executed, the high execution moves to the state \sS\ and waits for the corresponding input item to be requested by the low execution. After the low execution executes the input instruction at line~\ref{example:DI:LEx:l2} and \VFALSE\ is received from \linecode{cL1}, the high execution is waken up. Next, the instructions at lines~\ref{example:DI:HEx:l3}, \ref{example:DI:HEx:l5}, \ref{example:DI:HEx:l7}, \ref{example:DI:HEx:l9}, and \ref{example:DI:HEx:l10} are executed.

The execution of the low execution \Progl{1} is the same as the execution of the low program in NI. The execution of \Progl{2} follows the same pattern as the execution of \Progl{1}, with the only difference that also the result of the output instruction at line~\ref{example:DI:UEx:l10} is ignored by \REDUCE.

The input consumed and the output generated by the enforcement mechanism, the local input and output queues of the high execution \Progl{0} and the low execution \Progl{1} are the same as the ones in the enforcement mechanism of NI in Fig.~\ref{fig:example:execution:NI:IO}. The local input and output queues of the other local execution \Progl{2} are the same as respectively the input and output queues of the low execution.



\section{Soundness} \label{sec:soundness}
In this section we formalize the soundness property of an enforcement mechanism, and postulate the theorem on the security guarantees that the enforcement mechanisms configured earlier ensure with respect to the corresponding properties.

\begin{definition}\label{def:soundness}
An enforcement mechanism is \emph{sound} with respect to a property $P$ if for all programs $\pi$ the enforcement mechanism executed on $\pi$ satisfies $P$.
\end{definition}

\begin{theorem}[Soundness of Enforcement] \label{thm:soundess}
	Each enforcement mechanism in Tab.~\ref{tab:component:ifp} is sound with respect to the corresponding property, except for TSNI.
\end{theorem}
In order to prove soundness, we state two basic properties specifying the behaviour of \MAP\ on receiving requests for low input items and the behaviour of \REDUCE\ when receiving an output request from a local execution. For the actual proof of the theorem, we perform a case-based reasoning showing that at the end, the output is what we expect. In this respect an important assumption is that the program to be enforced is deterministic.  For RI, we need an additional preliminary property showing the relationship between the execution of a controlled program and the corresponding local execution. For DI, we need another preliminary property about the input items that can be consumed by \Progl{i} where $i > 1$. Figure~\ref{fig:proof:SN} sketches the proof strategy for soundness.

\begin{figure}
\centering
\includegraphics[scale=1]{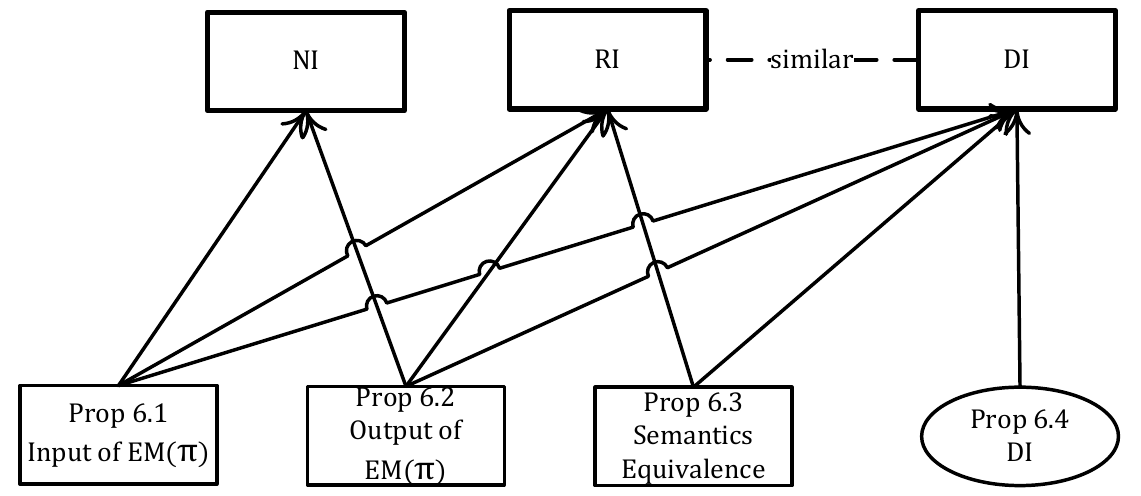}
\caption{Proof Strategy for Soundness}
\label{fig:proof:SN}
\end{figure}

We first prove the soundness theorem for NI. To this extend, we first need two simple propositions on the I/O behaviour of the enforcement components.

\begin{proposition}[Input items consumed by enforcement mechanisms and input items sent to local input queues of local executions]\label{prop:InputOfEMP} Consider the enforcement mechanisms of information flow properties, it follows that:
\begin{itemize}
\item \MAP\ will only ask low input items from the environment for low input requests from the low execution. The low execution can only receive default values for high input items. This local execution can receive real values for low input items.

\item For the enforcement mechanism of NI and DI, \MAP\ will only ask high input items from the environment for high input requests from the high execution.

\item For the enforcement mechanism of RI, \MAP\ will ask high input items from the environment for high input request from any local execution.

\item The high execution can receive real values for low and high input items.
\end{itemize}
\end{proposition}
\begin{proof}
The proposition is obvious from the configurations of the corresponding enforcement mechanisms, where all input items in local input queues are sent by \MAP; and only \MAP\ can get input items from the environment.

The proposition is proven by using the induction technique on the number of times of activation of \MAP\ on  input requests from local executions. Let $k$ be the number of times of activation of \MAP\ on the input request.

\textbf{Base case:} $k = 0$. The proposition holds vacuously.

\textbf{Induction hypothesis:} Assume that the proposition holds for the case that $k < n$. We now prove that the proposition holds for $k = n$. We consider the n-th activation of \MAP\ on a request from \Progl{i} on channel \chnl. The following holds:
\begin{itemize}
\item Each instruction of \MAP, \REDUCE, and local executions is executed atomically,
\item The execution of the \MAP\ program does not interfere with the execution of the \REDUCE\ program and vice-versa,
\item The execution of a local execution does not interfere with the execution of the \MAP\ (\REDUCE) program,
\end{itemize} Therefore, we have:

\begin{itemize}
\item Case 1: $i = 0$		
	\begin{itemize}
	\item Case 1.1: $\LVL[\chnl] = L$: For RI, the instructions at lines \ref{alg:RI:MAP:l1} and \ref{alg:RI:MAP:l7} in Fig.~\ref{alg:MAP:RI} are executed. For NI, the instructions at lines \ref{alg:NI:MAP:l1} and \ref{alg:NI:MAP:l11} in Fig.~\ref{alg:MAP:NI:SME} are executed. For DI, the instructions at lines \ref{alg:DI:MAP:l1}, \ref{alg:DI:MAP:l3} and \ref{alg:DI:MAP:l13} in Fig.~\ref{alg:MAP:DI} are executed. \MAP\ does not perform any input action. This activation does not influence the items received by the low execution.
	\item Case 1.2: $\LVL[\chnl] = H$: For RI, the instructions from line \ref{alg:RI:MAP:l1} to \ref{alg:RI:MAP:l5} in Fig.~\ref{alg:MAP:RI} are executed. For NI, the instructions from line \ref{alg:NI:MAP:l1} to \ref{alg:NI:MAP:l6} in Fig.~\ref{alg:MAP:NI:SME} are executed. For DI, the instructions from line \ref{alg:DI:MAP:l1} to \ref{alg:DI:MAP:l8} in Fig.~\ref{alg:MAP:DI} are executed. \MAP\ performs an input action and sends a default value to the local input queue of the low execution, and the real value to the local input queue of the high execution.
	\end{itemize}

\item Case 2: $i = 1$.
	\begin{itemize}
	\item Case 2.1: $\LVL[\chnl] = L$: The instructions executed are the same as the ones in Case 1.2, except for the clone instruction at line~\ref{alg:DI:MAP:l2} in Fig.~\ref{alg:MAP:DI} that is not executed. The real value is sent to the local input queues of  \Progl{1} and \Progl{0}.
	
	\item Case 2.2: $\LVL[\chnl] = H$: We check the instructions of the corresponding \MAP\ programs. \MAP\ does not perform an input action on low channels in this case. For NI, (respectively DI), the default value is sent to the local input queue of \Progl{1} by the execution of the instruction at line \ref{alg:NI:MAP:l8} in Fig.~\ref{alg:MAP:NI:SME} (resp. line~\ref{alg:DI:MAP:l10} in Fig.~\ref{alg:MAP:DI}). For RI, an input operation is performed (line \ref{alg:RI:MAP:l2} in Fig.~\ref{alg:MAP:RI}). However, a default value is sent to the local input queue of \Progl{1} (line \ref{alg:RI:MAP:l4}), while the real value is sent to the local input queue of \Progl{0} (line \ref{alg:RI:MAP:l3}).
	\end{itemize}
	
\item Case 3: $i > 1$ (only for the enforcement mechanism of DI)
	\begin{itemize}
	\item Case 3.1: $\LVL[\chnl] = L$: \MAP\ does not perform any input actions. \MAP\ does not send any input item to the local input queue (line \ref{alg:DI:MAP:l10} in Fig.~\ref{alg:MAP:DI}).
	\item Case 3.2: $\LVL[\chnl] = H$: \MAP\ will send only a default input item to the local input queue of \Progl{i} (line \ref{alg:DI:MAP:l10} in Fig.~\ref{alg:MAP:DI}).
	\end{itemize}

\end{itemize}

The proposition holds for $k = n$. Therefore, the proposition holds for all $k \geq 0$.
\end{proof}

\begin{proposition}[Outputs of enforcement mechanisms]\label{prop:outputOfEMP}
Concerning the output of an enforcement mechanism for an information flow property, it follows that:
\begin{itemize}
\item For the enforcement mechanism of RI, NI, and DI, only the high execution \Progl{0} sends output items to high output channels.
\item For the enforcement mechanism of RI, NI,  and DI, only the low execution \Progl{1} can send output items to low output channels.
\item For the enforcement mechanism of DI, the output items generated by the local execution \Progl{i} with $i > 1$ are ignored.
\end{itemize}
\end{proposition}
\begin{proof}
The proposition is proven by using the induction technique on the number of times of activation of \REDUCE\ on  output requests from local executions. The proof is similar to the proof of Prop.~\ref{prop:InputOfEMP}.
\end{proof}


Now we proceed to the proof of soundness of the enforcement mechanism for NI.
\paragraph{Proof of Theorem~\ref{thm:soundess} for NI.}

Let us consider two executions: \execution{\EMP}{I}{O} and \execution{\EMP}{I'}{O'}, where \loweq{I}{I'}.  

The following holds:

\begin{enumerate}
\item The low input items consumed by the enforcement mechanism depends only on the low execution (by Prop.~\ref{prop:InputOfEMP}). \label{lowEx:lowInput}

\item The low executions in the runs of the enforcement mechanism on $I$ and $I'$ always consume default values for high input items (by Prop~\ref{prop:InputOfEMP}). \label{lowEx:highInput}

\item  The low executions in these two runs consume the same low input items and same high output items. (By \ref{lowEx:lowInput}, \ref{lowEx:highInput} and \Prog\ be deterministic).  \label{lowEx:sameInput}

\item These low executions generate the same outputs (By \ref{lowEx:sameInput} and the fact that \Prog\ is deterministic). \label{lowEx:sameOutput}

\item The output items sent to low output channels are always generated by the low executions (by Prop~\ref{prop:outputOfEMP}). \label{output}

\item \loweq{O}{O'} (by \ref{lowEx:sameOutput} and \ref{output})
\end{enumerate}
This concludes the proof.\qed

For the proof of the soundness of the enforcement mechanism for RI, we need an  property stating the relationship between the controlled program and a local execution. We also need another simple property showing how \MAP\ handles the high input requests from local executions.

\begin{proposition}[Controlled programs and local executions] \label{prop:semantics-eq} Let $I_1$ and $I_2$ be two input queues, such that for all input channels \chnl, $I_1|_{\chnl} = I_2|_{\chnl}$. Then we have: for all programs \Prog,
\begin{multline*}
\forall I_1: (\Prog, I_1, \emptyQ) \rightarrowtriangle^k (\Prog_k, I_{1_k}, O_k) \implies \forall I_2:\forall \chnl \in \Cin: I_1|_{\chnl} = I_2|_{\chnl}: (\Prog, I_2,\emptyQ) \Rightarrow^k (\Prog, I_{2_k}, O_k)
\end{multline*} and $\channeleq{I_{1_k}}{I_{2_k}}{\chnl}$ for all \chnl.

And we have:
\begin{multline*}
\forall I_1: (\Prog, I_1,\emptyQ) \Rightarrow^k (\Prog_k, I_{1_k}, O_k) \implies \exists I_2:\forall \chnl \in \Cin: I_1|_{\chnl} = I_2|_{\chnl}: (\Prog, I_2,\emptyQ) \rightarrowtriangle^k (\Prog, I_{2_k}, O_k)
\end{multline*} and \channeleq{I_{1_k}}{I_{2_k}}{\chnl} for all \chnl.
\end{proposition}

\begin{proof}
The proposition is proven by using the induction technique on $k$ and the length of the input queue $I_1$, along with the fact that controlled programs and the local executions are deterministic.
\end{proof}



Next, we prove the soundness of the enforcement mechanism for RI.
\paragraph{Proof of Theorem~\ref{thm:soundess} for RI.}
Let \Irc{1} be the input consumed by the low execution in the run \execution{\EMP}{I}{O}, and $k$ be the number of high input items in $I$.

\textbf{Base case:} $k = 0$. The theorem holds for RI.(in this case $I' = I$).

\textbf{Induction hypothesis:} Assume that the theorem holds for $k < n$ for RI. We now prove that the theorem holds for $k = n$. The following holds:

\begin{enumerate}
\item \loweq{\Irc{1}}{I} (by Prop.~\ref{prop:InputOfEMP}). \label{RI:step1:loweq1}

\item \defseq{\Irc{1}|_H} (by Prop.~\ref{prop:InputOfEMP}). \label{RI:step2:HI}

\item There exists $I^*$ such that \execution{\Prog}{I^*}{O^*}, and  \channeleq{I^*}{\Irc{1}}{\chnl} for all \chnl\ (by Prop.~\ref{prop:semantics-eq}). \label{RI:constructedI}
\item \loweq{I^*}{\Irc{1}} (by \ref{RI:constructedI}, Prop.~\ref{prop:InputOfEMP}, and  the fact that \Prog\ is deterministic). \label{RI:step4:loweq2}
\item \loweq{I^*}{I} (by \ref{RI:step1:loweq1} and \ref{RI:step4:loweq2}).
\end{enumerate}

If we used $I^*$ as an input for the enforcement mechanism and the high execution is run only after the low execution is terminated:
\begin{enumerate}
\setcounter{enumi}{5}
\item The low execution will consume all input items in $I^*$ and is not stuck (by Prop~\ref{prop:InputOfEMP}). \label{RI:step6:lowEx}

\item Both the high and the low executions will consume with the same input (by \ref{RI:step2:HI}, \ref{RI:constructedI}, and Prop~\ref{prop:InputOfEMP}). \label{RI:step7:inputconsumed}

\item \label{termination} Both the high and the low execution are terminated (by \ref{RI:step6:lowEx} and \ref{RI:step7:inputconsumed}). \label{RI:step8:termination}

\item The output items are generated by the low execution (by Prop.~\ref{prop:outputOfEMP}). \label{RI:step9:outputgenerated}
\end{enumerate}

From \ref{RI:step4:loweq2}, \ref{RI:step8:termination}, and \ref{RI:step9:outputgenerated}, we have there exists $I^*$, such that \execution{\EMP}{I^*}{O^*} and \loweq{O^*}{O}. Let $I' \triangleq I^*$, it follows that the theorem for the enforcement mechanism of RI holds for the case $k = n$. Thus, the theorem holds for RI. \qed

To prove the soundness of DI, we need a proposition showing the influence of local execution \Progl{i} (with $i > 1$) on the input consumed by the enforcement mechanism.


\begin{proposition} \label{prop:DI:inputconsumed}
For the enforcement mechanism of DI, a local execution \Progl{i} with $i > 1$ has no effect on the input consumed by the enforcement mechanism.
\end{proposition}
\begin{proof}
The proof is obvious from the configuration of the enforcement mechanism.
\end{proof}

\paragraph{Proof of Theorem~\ref{thm:soundess} for DI.}
The idea of the enforcement mechanism of DI is that when the high execution requests a high input item, the high execution will be duplicated and the newly duplicated execution will receive the default values for high input items. If we replace the last high input item in the original input $I$ with a default item, then there exists another input queue satisfying the definition of DI. Such an input queue is the input consumed by the \Progl{\TOP}.


Let $I$ be an input, such that \execution{\EMP}{I}{O}. The proof of soundness of the enforcement mechanism of DI is based on the induction technique on the number of high input item in $I$.

\textbf{Base case:} If there is no high input item in $I$, the theorem holds vacuously.

\textbf{Induction Hypothesis}: The theorem holds for all $I$, such that the number of high input items is smaller than $n$. We now prove that the theorem also holds for the case when the number of high input items is equal to $n$. Now $I$ can be written as $I_1.\vec{v}_1.\dots.I_n.\vec{v}_n.I_{n+1}$

Based on the configuration of the enforcement mechanism, \Progl{TOP} is created when the high execution requests the last high input item. Let \Irc{\TOP} be the input consumed by \Progl{TOP}, \Irc{1} be the input consumed by \Progl{1}. We have:
\begin{enumerate}
\item $\Irc{\TOP} = I_1.\vec{v}_1.\dots.I_n.\defvec.I'_{n+1}$, where \defseq{I'_{n+1}|_H}. \label{DI:step:InputTOP}

\item \loweq{I}{\Irc{1}} \label{DI:step:InputLow}
\end{enumerate}

Let $I^*$ be an input queue, such that $I^* = I_1.\dots.I_n.I_{n+1}.\vec{v}_1.\dots.\vec{v}_{n-1}.\defvec.I^*_{n+1}$, where $I^*_{n+1} = I'_{n+1}|_H$. Assume that the order of executing local executions is first \Progl{1}, then \Progl{0}. We have:
\begin{enumerate}
\setcounter{enumi}{2}
\item \loweq{I^*}{I}. \label{DI:step:lowEqInput}

\item The low execution will consume the part $I_1.\dots.I_n.I_{n+1}$ for low input items and default values for high input items (by Prop~\ref{prop:InputOfEMP}). \label{DI:step:LowExInput}

\item \label{DI:step:HExInput} The high execution \Progl{0} will consume $\vec{v}_1.\dots\vec{v}_{n-1}.\defvec.I^*_{n+1}$ (by \ref{DI:step:InputTOP} and Prop.~\ref{prop:semantics-eq}).

\item For every $i$, such that  $1 < i \leq \TOP$, the execution of the local execution \Progl{i} is terminated and it does not effect the input consumed by the enforcement mechanism (by the assumption that \execution{\EMP}{I}{O} and Prop.~\ref{prop:DI:inputconsumed}). \label{DI:step:OtherEx}

\item $I^*$ is consumed completely by the enforcement mechanism (by \ref{DI:step:LowExInput} and \ref{DI:step:HExInput}). \label{DI:step:EMPInput}

\item The high execution is terminated (by the assumption that \execution{\EMP}{I}{O}). \label{DI:step:HExTer}

\item \execution{\EMP}{I^*}{O^*} (by \ref{DI:step:OtherEx}, \ref{DI:step:EMPInput}, \ref{DI:step:HExTer}). \label{DI:step:EMPTer}

\item \loweq{O^*}{O} (by Prop.~\ref{prop:outputOfEMP}) \label{DI:step:EMPOut}
\end{enumerate}

From \ref{DI:step:lowEqInput}, \ref{DI:step:EMPTer}, and \ref{DI:step:EMPOut}, the theorem holds for the case of the number of high input item in $I$ is $n$. Therefore, the theorem holds for the enforcement mechanism of DI. \qed

\section{Precision} \label{sec:precision}
The notion of precision for enforcement of a property is taken from \cite{Devr-Pies-10-IEEESP,DeGroef-etal-12-CCS}. The intuition is that the enforcement mechanism does not change the visible behavior of a program that is already secure with respect to the chosen property (and in particular each I/O on specific channels). Devriese and Piessens separated by construction the input queues of each channel. Since in our formulation the channels are merged into a global stream, our definition of precision must make explicit that the partial order of input items on a channel is preserved. This observation applies also to the order of output items in output queues. In our framework, the local executions are executed in parallel with no specific order. Therefore, the total order of input items consumed by the enforcement mechanism can be different from  the total order of input items in the input queue consumed by the controlled program that already obeys the desired property. However, the partial order of input items on a channel is preserved. This observation applies also to the order of output items in output queues.

\begin{definition}\label{def:precision}
An enforcement mechanism is \emph{precise} with respect to a property, if for any program \Prog\ that satisfies the property, and for every input $I$, where \execution{\Prog}{I}{O}, regardless of the order of executing local executions, the input $I^*$ and $O^*$ of the enforcement mechanism will be such that \channeleq{I^*}{I}{\chnl}, \channeleq{O^*}{O}{\chnl} for every channel \chnl, and \execution{\EMP}{I^*}{O^*}.
\end{definition}

\begin{theorem}[Precision of Enforcement] \label{thm:precision}
	Each enforcement mechanism in Tab.~\ref{tab:component:ifp} is precise with respect to the corresponding property, except for TINI.
\end{theorem}

\begin{figure}
\centering
\includegraphics[scale=0.8]{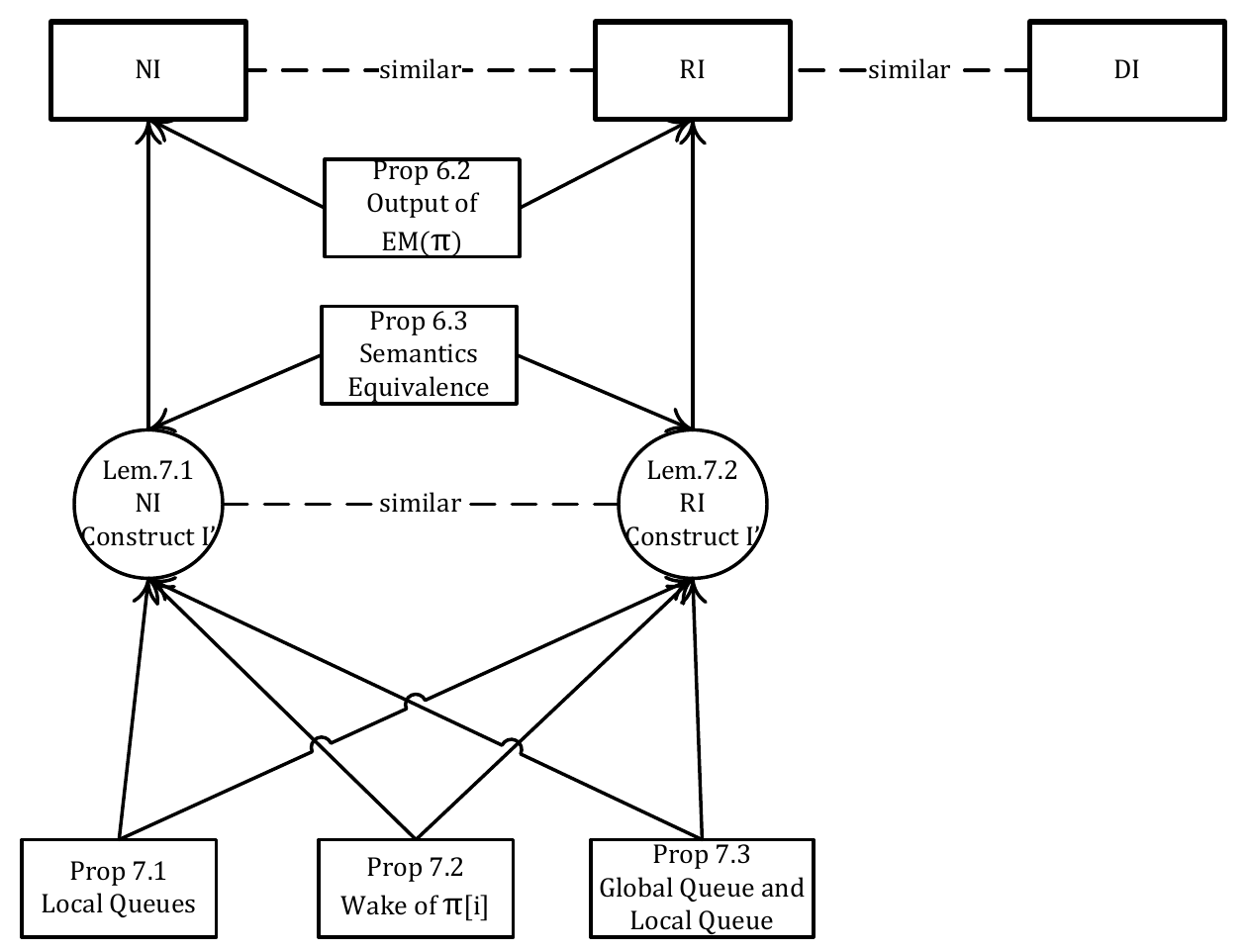}
\caption{Proof Strategy for Precision}
\label{fig:proof:PC}
\end{figure}

Figure~\ref{fig:proof:PC} shows the proof strategy for precision.
The proof of precision is more complex than the proof of soundness. At first, we need to prove a number of simple properties on the correct handling of interrupt signals and the equivalence between the semantics of controlled programs and the semantics of local executions.

\begin{proposition}[Local executions and local input queues] \label{prop:reuse:local}
For a local execution, when the input instruction is executed, if the input item required is in its local input queue, this item will be consumed. Otherwise, an interrupt signal is generated.
\end{proposition}
\begin{proof}
Proof follows obviously from the semantics of local executions.
\end{proof}

\begin{proposition}[The wake of local executions]
\label{prop:reuse:global}
The following facts hold:
\begin{enumerate}
\item If a local execution is sleeping on an input instruction that required an input item from the channel \chnl, this local execution will be waken up when the input item is ready and the instruction of \ProgM\ executed is the \NWAKE\ instruction. In addition, when a local execution is awaken, there is no interrupt signal in its configuration.

\item A local execution is not awaken when the input item required is not ready or when the input item required is ready, but the instruction executed of \ProgM\ is not the \NWAKE\ instruction.

\end{enumerate}
\end{proposition}

\begin{proof}
Proof follows by induction on the length of the derivation sequence of the enforcement mechanism.
\end{proof}

Next we show that from \Progl{0}'s input, we can reconstruct the original global input.

\begin{proposition}[Global input and local inputs] \label{prop:relationship:localinputqueue}
Let $k$ be the number steps of derivation of the execution of the enforcement mechanism of RI, DI, or NI. Assume that we have $(\EMP, I, \emptyQ) \Rightarrow^k (\EMP_k, I_k, O_k)$, and \Irck{0}{k} is the queue of the input items that have been received by \Progl{0}, then it follows that: \begin{itemize}
\item $\Irck{0}{k}.I_k = I$
\end{itemize}
\end{proposition}
\begin{proof}
The lemma is proven by using the induction technique on the length of the global input queue and the length of the derivation sequence of the enforcement mechanism, along with the fact that the execution of the controlled program and the executions of local executions are deterministic.
\end{proof}

At this point, we have all that is needed to present the key lemma for the proof of precision for NI that shows that all inputs have been processed and there is nothing left within the enforcement mechanism.

\begin{lemma}[Inputs of a controlled program and inputs consumed by the corresponding enforcement mechanism]\label{lem:ni:running}
Let \Prog\ be a program satisfying TSNI and \execution{\Prog}{I}{O}. Regardless of the order of executing local executions, if the low execution consumes the same low input items as in $I$, and the high execution consumes high input and low input items as in $I$, then it follows that the execution of the enforcement mechanism is terminated, and the input consumed by the enforcement mechanism is $I^*$, where \channeleq{I^*}{I}{\chnl} for all \chnl.


\end{lemma}

\begin{proof}
The proof of this lemma is based on the proposition of equivalence between semantics of controlled programs and semantics of local executions (Prop.~\ref{prop:semantics-eq}) and the proposition of the relationships between the global input queue and local input queues (Prop.~\ref{prop:relationship:localinputqueue}).

According to the semantics of the enforcement mechanism of NI, the high execution does not influence the termination of the low execution, the input consumed and the output generated by the low execution.

Therefore, regardless of the order of executing local executions, if the low execution consumes the same low input items as in $I$, then the input consumed by the low execution is $I|_L.(\defvec)^*.I_a$, where $I_a$ contains only low input items. We next prove that $I_a =\emptyQ$ and the low execution is terminated.

\begin{itemize}
\item Assume that $I_a \neq \emptyQ$. This means there exists an input $I'$, where $I'|_L = I|_L.I_a$ and \defseq{I'|_H}, and \executionnt{\Prog}{I'}. Since \Prog\ satisfies TSNI, this case cannot happen.

\item Assume that \Progl{1} is not terminated. However, this leads to the conclusion that \Prog\ does not satisfy TSNI.
\end{itemize}

We now prove that the high execution is also terminated and does not request any high input item that is not in $I$. 
\begin{itemize}
    \item Case 1: Assume that the high execution is stuck on a request for low input items. If the high input execution needs a low input item, the enforcement mechanism will behave accordingly to Prop.~\ref{prop:reuse:local} and Prop.~\ref{prop:reuse:global}. The high execution is stuck on low input items when it requests for an input item that is never requested by the low execution. Since the low input items consumed by the low execution is $I|_L$, the stuck of the high execution leads to the conclusion that \Prog\ is non-deterministic.

	\item Case 2: The high execution requests a high input items that is not in $I$. Regarding this assumption, because of Prop.~\ref{prop:semantics-eq}, there are two instances of \Prog\ that consume some input items, but at some point run in different paths of execution. In other words, \Prog\ is non-deterministic.

    \item Case 3: The high execution receives all input items it needs, but is in an infinite loop. This case also leads to the conclusion that \Prog\ is non-deterministic.
\end{itemize}

Therefore both local executions are terminated. Let \Irc{0} be the input queue received by \Progl{0}. Since \Progl{0} does not request any other input items that are not in $I$, then \channeleq{\Irc{0}}{I}{\chnl}. From Prop.~\ref{prop:relationship:localinputqueue}, we have $I^* = \Irc{0}$. Thus \channeleq{I^*}{I}{\chnl} and \execution{\EMP}{I^*}{O^*}.


\end{proof}

We have now all that is needed for the main theorem.

\paragraph{Proof of Theorem~\ref{thm:precision} for NI.}
Let $I$ be an input queue, such that $\execution{\Prog}{I}{O}$. We need to prove that regardless of the order of executing local execution, the input $I^*$ and output $O^^*$ will be such that \channeleq{I^*}{I}{\chnl}, \channeleq{O^*}{O}{\chnl}, and \execution{\EMP}{I^*}{O^*}.


The proof of precision of the enforcement mechanism of NI is based on Lem.~\ref{lem:ni:running} and Prop.~\ref{prop:outputOfEMP}. We have:
\begin{itemize}
\item Regardless of the order of executing local execution, the input $I^*$ and output $O^*$ will be such that \channeleq{I^*}{I}{\chnl}, and \execution{\EMP}{I^*}{O^*} (by Lem.~\ref{lem:ni:running}).
\item \channeleq{O^*}{O}{\chnl} (by Prop.~\ref{prop:outputOfEMP} and \Prog\ satisfying TSNI).
\end{itemize}

Therefore, the theorem holds for the enforcement mechanism of NI. \qed

For the RI property, the structure of the proof is similar. We can directly state the main lemma since we have already stated the key propositions.

\begin{lemma} \label{lem:ri:running}
Let \Prog\ be a program satisfying RI and \execution{\Prog}{I}{O}. Regardless of the order of running local executions, if the low execution consumes the same low input items as in $I$, and the high execution consumes high input and low input items as in $I$, then it follows that the execution of the enforcement mechanism is terminated, and the input consumed by the enforcement mechanism is $I^*$, where \channeleq{I^*}{I}{\chnl} for all \chnl.
\end{lemma}
\begin{proof}
The proof is similar to the proof of Lem.~\ref{lem:ni:running}.
According to the semantics of the enforcement mechanism of RI, the high execution does not influence the termination of the low execution and the output generated by the low execution. The high execution also does not influence the input consumed by the low execution, since \MAP\ can ask all input items for the low execution,  and the low execution only consumes default high input items.

Therefore, regardless of the order of executing local executions, if the low execution consumes the same low input items as in $I$, then the input consumed by the low execution is $I|_L.(\defvec)^*.I_a$, where $I_a$ contains only low input items.

The proofs that $I_a =\emptyQ$, the low execution is terminated, the high execution does not request any other low input items not in $I$ are similar to the proofs in Lem.~\ref{lem:ni:running}.

We now prove that all the high input items in \Irc{0} are consumed by \Progl{0}. Assume that there exists a high input item $\vec{v}$ in \Irc{0} ($\vec{v}[\chnl] \neq \NIL$) that is not consumed. From Prop.~\ref{prop:reuse:global} and Prop.~\ref{prop:reuse:local}, the existence of such a high input item means that the low execution requested a high input item that was not required by \Progl{0}. In other words, $\length{\restrict{\Irc{1}}{\chnl}} \not\leq \length{\restrict{\Irc{0}}{\chnl}}$. However, since \Prog\ satisfies RI, this case cannot happen.

Let \Irc{0} be the input queue received by \Progl{0}. Since \Progl{0} does not request any other input items that are not in $I$, then \channeleq{\Irc{0}}{I}{\chnl}. From Prop.~\ref{prop:relationship:localinputqueue},  we have $I^* = \Irc{0}$. Thus \channeleq{I^*}{I}{\chnl} and \execution{\EMP}{I^*}{O^*}.
\end{proof}

\paragraph{Proof of Theorem~\ref{thm:precision} for RI.}
We next prove the precision of the enforcement mechanism of RI. Let \Prog\ be a program satisfying RI, and $I$ be an input queue, such that $\execution{\Prog}{I}{O}$. We have:
\begin{itemize}
\item Regardless of the order of executing local execution, the input $I^*$ and output $O^*$ will be such that \channeleq{I^*}{I}{\chnl}, and \execution{\EMP}{I^*}{O^*} (by Lem.~\ref{lem:ri:running}).
\item \channeleq{O^*}{O}{\chnl} (by Prop.~\ref{prop:outputOfEMP} and \Prog\ satisfying RI).
\end{itemize}

Therefore, the theorem holds for the case of RI. \qed

\paragraph{Proof of Theorem~\ref{thm:precision} for DI.}
The proof for DI follows the same structure.


\section{Further Properties}\label{sec:discussion}

Our framework can capture other properties. Other BSPs from \cite{MANT-00-CSF} can also be enforced. Removal of events (RE) requires that if there is no high input, there is no high output. To enforce RE, when receiving an output request for a high channel from the high execution, \REDUCE\ needs to check whether there are any other high input items different from the default values and affecting the output generated by the high execution. Enforcement of strict removal of inputs (SRI) is similar to the enforcement of RI, but only the low execution can generate output items for both high and low channels. Strict deletion of inputs, deletion of events, and backward strict deletion can be enforced by using the \NCLONE\ instruction and the \REDUCE\ check mentioned above.

In \cite{Suth-86-SP} Sutherland defines the notion of \emph{non-deducibility} (ND) under the assumption that attackers have knowledge about the program, i.e. they know all possible executions of the program. ND can be enforced in our framework by running three local executions: the low, the high and the normal execution. Configuration of the low execution is similar to the low execution in the enforcement mechanism for NI. The normal execution will be privileged to use only input items read by other executions, and to output to high output channels. The high execution will read high input items and consume default low input items. In this way the attackers cannot deduce which execution occurred since there are other executions that can generate the same low behaviour. If an attacker based on his observations tries to construct a set of all possible executions that are low-equivalent, he cannot deduce which sequence of high events did not occur since the set he constructed contains all possible sequences of high events.

\begin{figure}

\begin{minipage}{0.5\textwidth}
\centering
\subfloat[\MAP\ for \prop1\ for an input from \chnl\ from \Progl{i}]{
\label{alg:MAP:SubDI}
\fbox{
\hspace{-12pt}
\begin{minipage}{1\columnwidth}
\begin{algorithmic}[1]
    \IF {$a \in \TAV[i][\chnl]$}                           \label{alg:SubDI:MAP:l1}
        \STATE \iinput{x}{\chnl} 						   \label{alg:SubDI:MAP:l2}
        \STATE \imap{x}{\chnl}{\canMap{\chnl}}			   \label{alg:SubDI:MAP:l3}
        \STATE \imap{\defVal}{\chnl}{\neg \canMap{\chnl}}  \label{alg:SubDI:MAP:l4}
        \STATE \iwake{\isReady{\chnl}}					   \label{alg:SubDI:MAP:l5}
    \ELSE												   \label{alg:SubDI:MAP:l6}
        \STATE \iskip									   \label{alg:SubDI:MAP:l7}
    \ENDIF
\end{algorithmic}
\end{minipage}
}
}\\
\subfloat[\TAV]{
\label{fig:table:SubDI:MAP}
\begin{minipage}{\columnwidth}
\begin{tabular}{|c|c|c|}
\hline
~ & \tindex{0} & \tindex{1} \\
\hline
$\LVL[\chnl] = H$ & \tgetask & \tnoaction \\
\hline
$\LVL[\chnl] = L$ & \tget & \tgetask \\
\hline
\end{tabular}
\end{minipage}
}
\end{minipage} \hfill
\begin{lrbox}{\mylistingbox}%
\begin{minipage}{0.35\columnwidth}
\vspace{53pt}
\begin{javascript}
input h1 from cH1
if h1 then
    input l2 from cL2
    input h2 from cH2
else
    input h2 from cH2
    input l2 from cL2
\end{javascript}%
\end{minipage}%
\end{lrbox}%
\subfloat[A program that satisfies RI, but not \prop1]{
\label{fig:example:RInotQ}
\usebox{\mylistingbox}
}
\caption{The new property - \prop1}
\label{fig:subDI}
\end{figure}

By modifying the privileges of local executions or modifying the \MAP\ or \REDUCE\ programs, we can enforce new properties. A possible modification is shown in Fig.~\ref{fig:table:SubDI:MAP}, where the configuration of \TAV\ is the same as the configuration of \TAV\ for NI, and the \MAP\ program is the same as the one for RI. The low execution needs to wait for high input items requested by the high execution even though the low execution can only consume default values. This option leads to a novel strict property, which we have called \emph{substitution-deletion of inputs} (\prop1). \prop1\ requires that when all high input items in an input $I$ are substituted by a default item or deleted, then the remaining input items can be corrected to $I'$, which preserves the low prefixes of high input items in $I$. A program that satisfies RI, but not $\prop1$, is described in Fig.~\ref{fig:example:RInotQ}; therefore these properties are actually different.

\begin{figure}
\centering
\begin{tabular}{|c|c|c|}
\hline
~ &  \tindex{0} & \tindex{1} \\
\hline
$\LVL[\chnl] = H$ & \ttellput & \ttellput \\
\hline
$\LVL[\chnl] = L$ & \tnoaction & \ttellput \\
\hline
\end{tabular}
\caption{A configuration of \TPV, in which the low execution can send output to high channels}
\label{fig:table:NI:nowaitoption}
\end{figure}

The configuration of \TPV\ also leads to discovery of new properties. A possible configuration of \TPV\ is described in Fig.~\ref{fig:table:NI:nowaitoption} in which the low execution can send output items to high channels.

\section{Relationships among the properties enforced}\label{sec:discussion:relationship}
We have defined enforcement mechanisms for several information flow properties, but it might be unclear whether these properties are actually different (in our notation). Further we demonstrate that the properties that we have investigated are not the same.



\paragraph{The relationship between RI and NI.}\label{sec:discussion:relationship:rini}
The RI property is stricter than the property of TINI, because the RI property requires that if a real value is replaced by a default one, then the other real values are either also replaced by the default ones, or will not appear in the input squeue at all. Actually, if a program satisfies the RI property, then it also satisfies the NI property. However, the opposite is not true.

\begin{figure}
\centering
\begin{lrbox}{\mylistingbox}%
\begin{minipage}{0.37\textwidth}
\vspace{55pt}
\begin{javascript}
input x from cH1
input value from cL1
if x then
    input y from cH2
else
    input z from cH3
output value to cL2
\end{javascript}
\end{minipage}
\end{lrbox}
\subfloat[A program that satisfies NI, but not RI]{
\label{fig:example:counter:rini:source} \usebox{\mylistingbox}
} \hspace{10pt}
\begin{lrbox}{\mylistingbox}%
\begin{minipage}{0.37 \textwidth}
\begin{javascript}
input h1 from cH1
if h1 then
    input h2 from cH2
    if h2 then
        input l1 from cL1
    else
        input l1 from cL1
        while true do skip
else
    input h2 from cH2
    input l1 from cL1
output l1 to cL2
\end{javascript}
\end{minipage}
\end{lrbox}
\subfloat[A program that satisfies RI, but not DI]{
\label{fig:example:RInotDI}
\usebox{\mylistingbox}
}
\caption{Examples of programs satisfying a property, but not another}
\label{fig:example:counter:source}
\end{figure}

A program, which satisfies NI, but not RI, is shown in  Fig.~\ref{fig:example:counter:rini:source}. In this example the level of the channels \linecode{cH1, cH2} and \linecode{cH3} is high, while the level of the channels \linecode{cL1} and \linecode{cL2} is low. This program satisfies the NI property, since the output item on the channel \linecode{cL5} is independent from the secret values from confidential channels. However, this program does not satisfy the RI property. The reason is that the execution of the program with the input $(\linecode{cH1}=\VTRUE) (\linecode{cL1}=\valueM)(\linecode{cH2}=\valueM)$ is terminated, but if we apply the procedure of perturbation and correction on this input, the results are inputs with which the execution of the program will be error.

\paragraph{The relationship of DI and RI}
A program satisfying RI, but not DI, is shown in Fig.~\ref{fig:example:RInotDI}. In this example, $I = (\linecode{cH1}=\VTRUE)(\linecode{cH2}=\VTRUE)(\linecode{cL1}=\valueM)$ is an input and the execution of the program on $I$ generates the output $O = (\linecode{cL2}=\valueM)$. If we replace all high inputs with the default values, we obtain another input $(\linecode{cH1}=\VFALSE)(\linecode{cH2}=\VFALSE)(\linecode{cL1}=\valueM)$, such that the program in Fig.~\ref{fig:example:RInotDI} will produce an output low equivalent with $O$. However, if we replace the last high input event by the default value, then the new input is $I^*  = (\linecode{cH1}=\VTRUE)(\linecode{cH2}=\VFALSE)(\linecode{cL1}=\valueM)$, and the execution of the program with $I^*$ is not terminated. Therefore, the DI property does not hold for the program.

\section{Limitations}\label{sec:discussion:limitations}
Currently, the enforcement mechanism is not independent from the choice of the default values (\defVal). We prove soundness and precision of enforcement with respect to all possible choices of the default values, and we assume that for each channel it is possible to determine a suitable (``non-leaking'') default value.

The definition of DI in our notation requires that high items in $I_2'$ are default ones. This constraint is enough to prevent attackers from deducing whether the value of the last high input is default or not. However, we can put another constraint on $I_2'$, i.e, $\length{\restrict{I_2'}{\chnl}} \leq \length{\restrict{I_2}{\chnl}}$ for all \chnl. Regarding this additional constraint, the relationship between DI and RI as shown in \cite{MANT-00-CSF} is preserved.


Our enforcement mechanism in \S\ref{sec:em:NI} inherits the limitations of the SME mechanism \cite{Devr-Pies-10-IEEESP}. SME can soundly enforce TINI, but not TSNI. This happens in the case when the low execution is terminated but the high execution is not, and thus the whole enforcement mechanism is not terminated. Respectively, SME (and our enforcement mechanism for NI) can precisely enforce TSNI, but not TINI.




In \cite{Kash-Wied-Hard-11-SSP} Kashyap, Wiedermann and Hardekopf evaluate the security guarantees of SME for the termination covert channel; they have proposed to mediate the security problems of SME related to this channel with more sophisticated schedulers. In our approach we do not schedule the order of local executions, therefore, we cannot immediately adapt their suggestions. However, our framework can be extended to control the order of executing local executions by specifying a new rule to control the start of local executions, and the predicate \isReady{} used in the \NWAKE\ instruction.

%
%

 We see one of the main limitations of our current proposal in the absence of a practical implementation. It is still an open question, whether the memory and performance overhead will be acceptable, especially for complex properties, such as DI. Devriese and Piessens in the original SME paper \cite{Devr-Pies-10-IEEESP}, as well as Bielova et al. in \cite{Biel-etal-11-NSS} and De Groef et al. in \cite{DeGroef-etal-12-CCS} report on complications while instrumenting SME for real browsers, which we will have to address. A working implementation is our next target.

\section{Related Work}\label{sec:relwork}
The information flow policies enforcement is a deeply investigated field. We will briefly recall the developed approaches for information flow policies enforcement and discuss the most relevant techniques in more details.

Static analysis techniques for information flow security inspect the program code in order to check whether there is any unwanted information flow. We refer the interested reader to the survey by Sabelfeld and Myers \cite{Sabe-Myer-2003} with an excellent overview of static language-based approaches for information flow security.




In contrast to the static verification techniques, dynamic analysis for information flow enforcement tracks propagation of confidential information when a program is executed;  an extensive review on the dynamic approach can be found in \cite{LeGu-07}. The trade-offs between static and dynamic analysis approaches are evaluated by Russo and Sabelfeld in \cite{Russ-Sabe-10-CSF}. 




Our choice of using the multi-execution approach, despite its performance overhead, was dictated by its advantages over the static and dynamic information flow analysis techniques. Static analysis can fall short in scenarios when the program can be composed dynamically, such as in the case of JavaScript; dynamic runtime monitoring for information flow can can suffer from impossibility to account for the branch of execution that was not taken, and can leak the control flow details \cite{Sabe-Myer-2003}, e.g. introduce information leaks through the halting behaviour of the program. 

Secure multi-execution \cite{Devr-Pies-10-IEEESP} has inspired many researchers to push further investigation of this technique. Jaskelioff and Russo in \cite{Jask-Russ-2012} describe their adaptation of SME to Haskell and provide an SME implementation in a handy library. SME is applied to a reactive model of a browser in \cite{Biel-etal-11-NSS}, and is implemented as a fully functional web browser FlowFox that embeds an SME-based runtime enforcement mechanism  in \cite{DeGroef-etal-12-CCS}. FlowFox is a modification of Firefox, it introduces a noticeable memory and performance overhead, but works with most of the existing web sites. We plan to learn from \cite{DeGroef-etal-12-CCS} how to implement a fully working solution and how to evaluate the usability.

Barthe et al. \cite{Bart-Crespo-FMOODS-2012} achieve the effects of SME through static program transformation instead of modifying the runtime environment. The transformation technique, which provides sound and precise enforcement of non-interference, is based on the main SME idea: a program is transformed into the sequential composition of the same code paired with a security level ranging from low to high. The program instances on higher levels reuse the inputs of the low instances through global buffers of inputs. We achieve the same security with our enforcement mechanism configuration for NI.

In \cite{Kash-Wied-Hard-11-SSP} the authors analyse SME with respect to timing and termination covert channels and propose a variety of schedulers for SME to close these channels. Our framework can be extended with a scheduler to orchestrate the local executions order of executions; we plan to apply the subtleties of timing- and termination-sensitive non-interference identified by the authors and improve our framework by closing these covert channels.

Instead of having multiple different executions, in \cite{Aust-Flan-12-POPL} non-interference is achieved by using faceted values (pairs of two values containing low and high information). This allows to effectively simulate multiple executions on different security levels while in fact running a single-process execution (the projection property).   The authors also introduce enforcement of declassification policies with their technique. Our enforcement mechanism for NI ensures the same security as the faceted values approach; however, our framework currently does not include support for declassification. Yet, enforcement of such property as DI seems infeasible for the faceted values approach in its current state, as significant changes in the semantics for treatment of faceted values are required. 

Capizzi et al. in \cite{Capi-Long-Venk-08-ACSAC} propose the shadow execution technique, with the main goal to prevent confidentiality leaks of the user information in an operating system setting. Shadow execution, which idea is similar to SME, consists of replacing the original program with two copies. The private copy (the high execution) receives the confidential data, but is prevented from accessing the network. The public copy (the low execution) receives fake data, but can access the network; the results from the network are supplied also to the private copy. In this way the private copy can avail any network related functionality without leaking the confidential data. Our framework can be configured to fully simulate the shadow execution technique, by using the configuration for enforcement of NI and regulating the network connectivity of the program copies.

With respect to the body of work on the SME technique, we do not push further the security guarantees offered by the SME paper \cite{Devr-Pies-10-IEEESP}; for the non-interference property we do not improve the SME drawbacks and do not close the reported covert channels. However, our goal is different. We aim at creating an enforcement framework which can be easily configured to accommodate different information flow properties; and multi-execution is just the technique we have chosen for achieving our goal. Other techniques (e.g. the faceted values approach) can also be considered.




\section{Conclusion}\label{sec:conclusion}
We have presented an architecture of an extensible framework for enforcement of information flow properties. To the best of our knowledge, this is the first enforcement mechanism capable to accommodate more than one property. The main idea behind our approach is to run several local instances of a program in parallel, following the idea of secure multi-execution \cite{Devr-Pies-10-IEEESP}, and carefully orchestrate processing of input and output operations of the enforced program through two components (\MAP\ and \REDUCE), and two tables (\TAV\ and \TPV).

To support our claims on the extensibility of the framework we have provided a set of configurations of the enforcement framework for enforcement of non-interference and several properties from the framework of Mantel \cite{MANT-00-CSF}. The the framework components programs for each of these properties are quite simple. For these properties we have formally proven soundness and precision of enforcement.

Our approach to enforcement allows the users to define and enforce novel information flow properties, only by fine-tuning the settings. This characteristics of our framework has huge potential. We plan to continue studying new properties that appear due to fine-tuning the components of the framework, and to research configurations for combinations of several properties.

One may argue whether it is actually easy to configure the enforcement in order to add a new property. Of course, we do not expect that an average user might be interested in defining his own property; we only aim to provide such a user with a checkbox selection of a desired information flow property. However, more security-aware users and the security researchers with a desire to investigate some information flow property will be able to obtain the enforcement mechanism by customizing our framework with simple programs for \MAP\ and \REDUCE, and defining the tables \TAV\ and \TPV, instead of hacking their browsers.

Our next steps include the investigation the patterns of \TAV, \TPV, \ProgM, \ProgR\ and the property to be enforced; a proof-of-concept implementation (we have chosen to implement our framework for a web browser; however, our approach can be suitable to any platform), and extension of the framework with more properties and options for declassification.

\section*{Acknowledgment}
This work is partly supported by the projects EU-IST-NOE-NESSOS and EU-IST-IP-ANIKETOS.

\newpage

\newpage
\appendix



\end{document}